\documentclass[10pt,letterpaper]{IEEEtran}

\usepackage{amsmath}
\usepackage{amsthm}
\usepackage{amsfonts}
\usepackage{amssymb}
\usepackage{cite}
\usepackage{graphicx}
\usepackage{bbm}
\usepackage{url}

\newtheorem{theorem}{Theorem}
\newtheorem{corollary}{Corollary}
\newtheorem{proposition}{Proposition}
\newtheorem{definition}{Definition}
\newtheorem{remark}{Remark}

\DeclareMathOperator*{\argmin}{arg\,min}
\DeclareMathOperator{\cC}{\mathrm{C}}
\DeclareMathOperator{\cO}{\mathcal{O}}
\DeclareMathOperator{\cL}{\mathcal{L}}
\DeclareMathOperator{\bS}{\mathbb{S}}
\DeclareMathOperator{\bR}{\mathbb{R}}
\DeclareMathOperator{\bP}{\mathbf{P}}
\DeclareMathOperator{\bE}{\mathbf{E}}
\DeclareMathOperator{\ind}{\mathbbm{1}}
\DeclareMathOperator{\SIR}{\textrm{SIR}}
\newcommand*\diff{\mathop{}\!\mathrm{d}}
\newcommand*\nnb{\nonumber}

\title{Satellite Infrastructure Sharing:\\ Orbit-Structured Stochastic Geometry Modeling and Connectivity Analysis in Heterogeneous Satellite Networks}
\author{Chang-Sik Choi~\IEEEmembership{Member,~IEEE}
	\IEEEcompsocitemizethanks{\IEEEcompsocthanksitem{Chang-Sik Choi is an Assistant Professor at School of EE, KAIST, South Korea. email: changsik@kaist.ac.kr}}
}
\begin{document}

	\maketitle
\begin{abstract}
This paper develops an analytical framework to evaluate the feasibility and performance of satellite infrastructure sharing among multiple low Earth orbit (LEO) satellite operators. Motivated by the growing demand for universal connectivity under limited satellite resources, the proposed model captures uncoordinated deployments where independently operated constellations coexist without predefined orbital agreements. To describe such heterogeneous configurations, the spherical Cox point process is employed to jointly generate orbital structures and satellites. Then, each satellite is further assigned a random communication range, reflecting variations in coverage capability across operators. The overall coverage region is modeled through a spherical Cox-Boolean model that captures the spatial overlap of individual satellite spherical footprints on Earth. Using the proposed framework, the feasibility and benefits of satellite infrastructure sharing are mathematically analyzed, and closed-form expressions are derived for key performance metrics such as the connection probability, connection number, and downlink signal characteristics including the nearest serving distance, total received signal power, and the signal-to-interference ratio (SIR) distribution in the interference-limited regime. The analytical results, validated through system-level simulations, provide a tractable characterization of how orbital geometry governs coverage, connectivity, and interference, and reveal the inherent trade-offs induced by coverage overlap in heterogeneous satellite constellations.

\end{abstract}


\section{Introduction}

\IEEEPARstart{S}{atellite} networks are capable of providing universal coverage to mobile devices located anywhere on Earth~\cite{6934544,8700141,9275613,9861699,9941481}. By establishing connections between satellites, gateways, and mobile users, satellite infrastructure can deliver seamless service to user terminals on Earth~\cite{8002583}. In particular, for low Earth orbit (LEO) satellite constellations, the association time of ground-to-space links is relatively short due to the high orbital velocity of the satellites. Consequently, a large number of satellites are required to provide reliable and continuous services, especially when serving dense user populations. In practice, systems such as Starlink and Kuiper are planned to operate with tens of thousands of satellites to achieve near-global coverage, while others deploy only a limited number due to cost or regulatory constraints~\cite{FCCKuiper,FCCBoeing,FCCStarlink,9371230}.

However, achieving reliable universal coverage remains challenging due to the enormous deployment scale and the spatial and temporal nonuniformity of user demand~\cite{9210567,okla}. As satellite networks continue to expand, multiple operators are expected to coexist in overlapping orbital altitudes, often using distinct constellations with different coverage footprints~\cite{9210567,9275613,9852737}. This natural coexistence motivates the study of \emph{satellite infrastructure sharing}, where satellites operated by independent providers may collectively contribute to global connectivity. Understanding the quantitative benefit of such sharing—especially under random, uncoordinated orbital distributions—is essential for future non-terrestrial network (NTN) design~\cite{6934544,8700141,9275613,9861699}.

Despite the growing interest in the collaborative use of satellite resources, the quantitative impact of multi-operator satellite sharing on coverage performance remains largely unexplored. This paper addresses this gap by developing a tractable analytical model, based on stochastic geometry~\cite{daley2007introduction,chiu2013stochastic,baccelli2010stochastic,haenggi2012stochastic}, that characterizes the spatial distribution of shared constellations and evaluates the resulting spherical coverage regions formed on the Earth’s surface by satellites located on various orbits.

\subsection{Related Work}
Recent advances in NTNs and the rapid deployment of LEO and medium Earth orbit (MEO) constellations have significantly accelerated research on the modeling and performance evaluation of satellite communication systems. With the growing number of operators and constellations, the spatial distribution of satellites and their mutual interference have become critical factors determining network reliability and scalability. Numerous studies have sought to analyze the coverage, connectivity, and interference characteristics of such satellite systems by employing both deterministic orbital simulations and analytical models~\cite{baccelli2010stochastic,haenggi2012stochastic}. While deterministic analyses can capture precise orbital mechanics, they often lack tractability when evaluating large-scale, multi-constellation networks. Consequently, analytical approaches that can statistically characterize the spatial distribution of satellites and quantify typical network performance are essential tools for system-level understanding and design optimization~\cite{11159552}.

Stochastic geometry has been used extensively for the analysis of various network architectures including cellular networks~\cite{6042301,6171996}, ad hoc networks~\cite{5226963}, and vehicular networks~\cite{8340239,8357962,8419219,9354063}. Among existing stochastic geometry models for satellite communications, the binomial point process and the Poisson point process have been predominantly employed to characterize the random spatial distribution of satellites on a sphere~\cite{9079921,9177073,9218989}. The binomial model assumes a fixed number of satellites uniformly distributed over the spherical satellite surface, while the Poisson model captures uniform random deployments with variable population, providing a more flexible mathematical framework. Leveraging these models, numerous studies have analytically investigated the large-scale behavior of satellite networks, enabling the derivation of typical user performance without resorting to exhaustive system-level simulations. For instance, stochastic geometry has been applied for evaluating downlink communications under various fading conditions~\cite{9678973}, analyzing the coverage probability and inter-satellite interference~\cite{kim2024spectrum}, deriving non-orthogonal downlink communications~\cite{10901947}, and assessing the localization through satellites~\cite{10994488}.

In a related context, stochastic geometry has also been applied to airborne and UAV-based communication systems, where nodes are typically modeled over planar or finite regions using point processes such as the 3D binomial or Poisson point processes~\cite{7967745,8713514,8654644,8681266,9252889,9785498,10050345}. In these models, the coverage region of UAVs is commonly represented by disk-shaped regions determined by communication ranges or antenna patterns, enabling tractable analysis of coverage probability, typical interference, SIR distribution of the typical receiver, mobility pattern, or connectivity through airborne network elements. While these approaches provide useful insights into aerial networks, they do not directly extend to satellite systems. In particular, UAV-based models do not incorporate orbital constraints, and thus lack the structural dependence induced by satellites being confined to orbital planes. As a result, coverage and connectivity in UAV networks are primarily governed by spatial randomness within a finite region, whereas satellite systems exhibit orbit-induced spatial structure on a spherical surface. This distinction motivates the need for models that explicitly capture orbital geometry.

Although the binomial and Poisson point process models have been useful for analyzing the large-scale behavior of satellite communication performance, the orbital foundation and clustering of satellites cannot be captured by these random point processes because they lack explicit orbital structure \cite{10557592}. In light of these limitations, particularly the absence of orbit-induced structure highlighted above, the Cox point process on the spherical surface was developed by~\cite{10410220,10557592,10703111} to explicitly incorporate orbital geometry into the modeling framework. By simultaneously generating the orbits and then distributing satellites only along those orbital planes, the Cox point process successfully characterizes the underlying orbital structure of LEO and MEO constellations in a spatially consistent manner, exploiting the conditional structure of the developed Cox process. The Cox point process has since been used to model and analyze various applications, including downlink communications~\cite{10703111}, data harvesting or sensing architectures~\cite{10436110}, and new NTN designs leveraging aerial platforms \cite{10771991}. Some recent studies have also employed the Cox point process to derive the downlink signal-to-interference ratio (SIR) under a general fading assumption~\cite{subCox1}.

It has been shown that the Cox model can provide a useful approximation of existing and forthcoming satellite constellations~\cite{10410220,10703111}. The moment-matching method has proven to be a reliable tool for locally approximating practical satellite constellations with analytical point-process models such as the binomial or Poisson point processes~\cite{9079921,9177073,9218989}. Likewise, the Cox model employs this method to approximate a real constellation; however, a strong advantage of the Cox model over the binomial and Poisson models is that it not only approximates the satellite distribution itself but also replicates the orbital foundation  at the same time~\cite{10410220,10703111}. Consequently, the Cox model can approximate not only the total number of satellites in a system but also local geometric structure of the underlying orbit geometry. This flexibility is in particular crucial for representing multi-constellation heterogeneous satellite systems, where the orbital foundations are uncoordinated or even seemingly random across network operators. In such settings, the system becomes more random and isotropic, and the Cox model accurately captures both the global layout and the clustering of satellites along their orbits.

This paper builds on these insights by proposing a modeling framework to analyze satellite infrastructure sharing, where the clustering of satellites on their orbital planes plays a key role in determining network performance. To this end, we adopt the Cox point process to model and describe satellite networks composed of multiple operators whose orbital parameters are not coordinated or predetermined. Under this setup, we investigate the performance of satellite infrastructure sharing and reveal how the random and uncoordinated orbital foundations influence the achievable coverage and connectivity of the overall system by leveraging theories of random closed sets under the stochastic geometry framework~\cite{molchanov2005theory,chiu2013stochastic,baccelli2010stochastic}.

In particular, we characterize geometric connectivity metrics such as the connection probability and the connection number, which quantify coverage availability and redundancy from a typical user’s perspective. We further connect this geometric formulation to communication-level performance by analyzing the nearest serving distance, the total received signal power, and a baseline SIR under a nearest-association rule. To assess the practical relevance of the proposed framework, we also compare the resulting geometric metrics with those obtained from a down-sampled realistic LEO constellation, where the Cox parameters are calibrated using visibility statistics observed from the typical user. These results provide insight into how orbit-induced geometry shapes geometric coverage, coverage overlap, and downlink communications in heterogeneous satellite networks. To the best of the authors’ knowledge, this is the first work that models and systematically analyzes satellite infrastructure sharing in a unified framework based on geometric coverage cells over multiple orbits.

\subsection{Theoretical Contributions}

This paper makes the following contributions.

First, we develop an orbit-structured stochastic geometry framework for modeling heterogeneous satellite networks under infrastructure sharing. To capture the lack of coordination across multiple operators, we employ a Cox point process on the sphere, where orbital planes are generated in an isotropic manner and satellites are randomly distributed along these orbits. On top of this construction, each satellite is assigned an independent communication range, which induces a collection of spherical coverage cells on the Earth surface. The resulting structure forms a Boolean-type model driven by the Cox point process, enabling a systematic geometric representation of the total coverage region created by multiple uncoordinated constellations.

Second, we introduce and analyze geometric connectivity metrics that explicitly reflect the overlap structure induced by orbit-constrained deployments. In particular, we characterize the connection probability, defined as the probability that a typical user is covered by at least one satellite, and the connection number, defined as the number of satellites whose coverage cells include the user. The latter provides a direct measure of coverage redundancy and multi-connectivity potential under infrastructure sharing. By exploiting the conditional structure of the Cox model, we derive analytical expressions that characterize these metrics and reveal how orbit-induced clustering shapes the achievable connectivity.

Third, we establish a connection between the geometric coverage model and communication-level performance. In particular, we characterize the distribution of the nearest serving distance and the total received signal power observed at the typical user. The nearest distance distribution provides a direct geometric proxy for the received signal strength, while the total received signal power reflects the aggregate contribution of multiple covering satellites under infrastructure sharing. These metrics bridge the gap between coverage geometry and link-level performance in a tractable manner. We further consider a baseline SIR evaluation in which the serving satellite is selected as the nearest covering satellite, while the remaining covering satellites act as interferers. This allows us to examine how coverage overlap induced by orbital geometry translates into interference and affects link quality. In addition, we validate the geometric accuracy of the proposed model through a comparison with a down-sampled realistic LEO constellation, where the Cox parameters are calibrated using visibility statistics observed from the typical user. The comparison shows that the Cox model captures the overall trend and scale of key geometric metrics, supporting its use as a first-order approximation of realistic orbit-structured deployments.

Finally, the analytical framework provides explicit dependence of key performance metrics on system parameters such as the number of orbital planes, the number of satellites per orbit, and the communication range. This enables tractable evaluation of satellite infrastructure sharing in heterogeneous satellite networks by offering insights into overall coverage, geometric overlap, typical downlink communications.

\section{System Model}\label{S:2}
We consider heterogeneous satellite networks in which multiple network operators coexist, and numerous satellites are deployed independently by each operator. In the proposed satellite infrastructure sharing scenario, a network user on Earth can communicate freely with any satellite of any operator, provided that it lies within the satellite’s communication range. This range is determined by various factors, including the transmit power, orbital altitude, and antenna characteristics of the satellites. The specific assumption on the communication range will be detailed in Section~\ref{section:range}.
\subsection{Stochastic Geometry Model For Satellite Locations}

\begin{figure}
	\centering
	\includegraphics[width=1\linewidth]{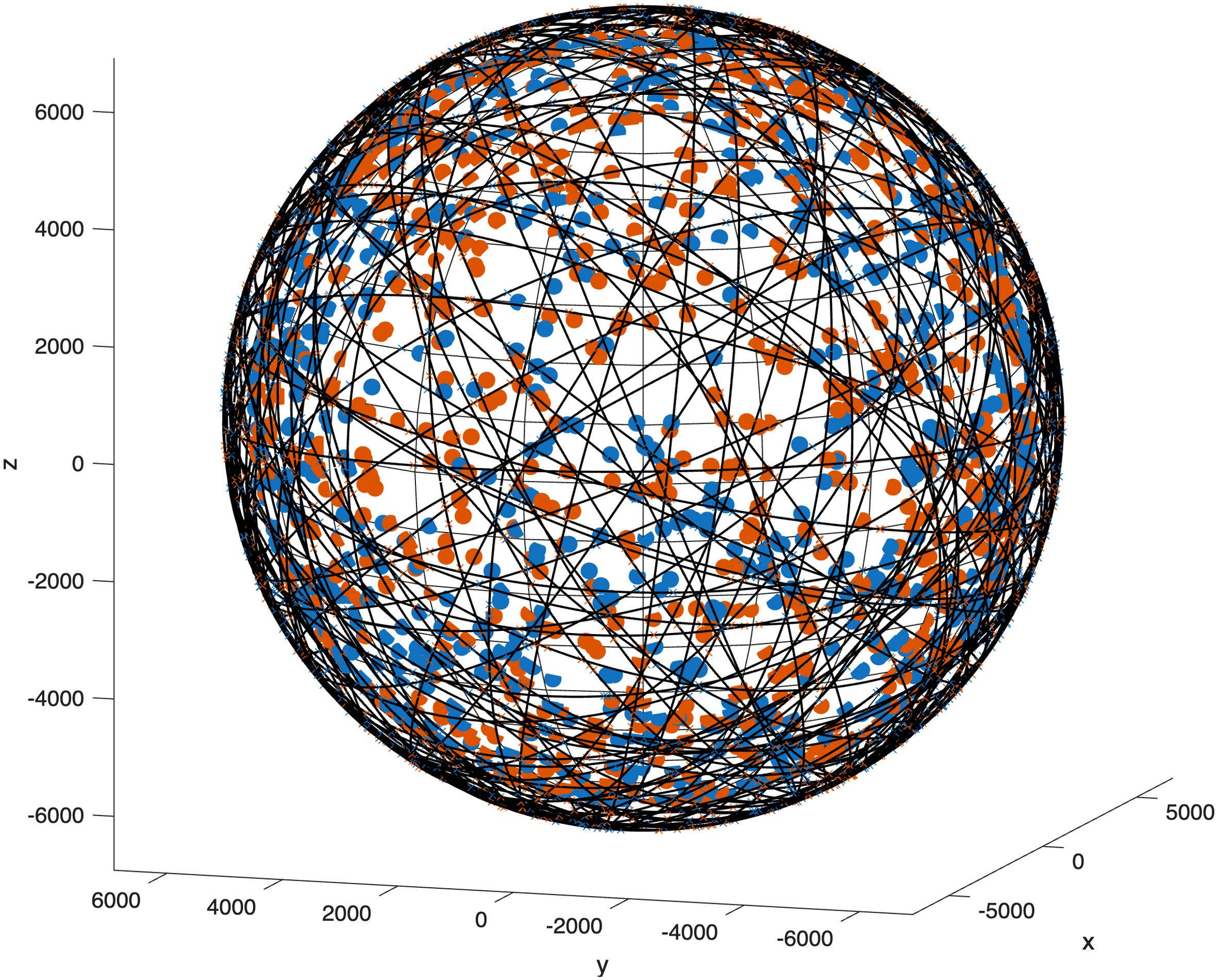}
	\caption{Illustration of satellite sharing with two operators. In total, $90 $ isotropic orbital planes and $30 $ satellites per orbit on average. The cells are described as circles on Earth sphere.}
	\label{fig:1}
\end{figure}

To accommodate the real-world observation that orbital foundations from various operators coexist without predefined coordination,\footnote{A low level of coordination must exist to prevent collisions, yet a high level of coordination aimed at maximizing joint coverage does not exist.} we adopt the spherical Cox point process model developed in~\cite{10410220,10557592,10703111}, in which both orbits and the satellites on them are isotropically distributed. The Cox model captures the underlying spatial structure of LEO and MEO satellite networks by reflecting the correlation of satellites along their orbits. Specifically, by generating orbital foundations and then populating satellites exclusively along these orbits, the model ensures that satellites are distributed in a spatially consistent manner.

Specifically, the isotropic orbit process $\cO$ is given by a Poisson point process $\Xi$ of intensity $\lambda\sin(\phi)/2\pi$ in a two-dimensional torus $\mathbb T:(\theta,\phi)= [0,\pi]\times [0,\pi]$. Then, each point of $\Xi$ with its coordinates  $(\theta,\phi)$ produces an unoriented orbit $O(\theta,\phi)$ in the Euclidean space $\bR^3$ with its radius $r_s$ (km), longitude $\theta$ (rad), and inclination $\phi$ (rad). The orbit process is denoted by
\begin{equation}
	\cO = \sum_{i} O_i = \sum_{i} O(\theta_i,\phi_i) .
\end{equation}
Conditionally on the orbit process, the phases of satellites on the $i$-th orbit are assumed to be distributed as a Poisson point process $\psi_{i}$ of intensity $\mu/2\pi$ on $[0,2\pi]$. The parameter $\mu$ is the expected number of the satellites on each orbit. The satellite point process on $i$-th orbit is denoted by $\psi_i = \sum_{j}\delta_{X_{i,j}},$
where $\delta_{X_{i,j}}$ denotes the Dirac measure at $X_{i,j}\in\bR^3$ and $X_{i,j}$ denotes the $j$-th satellite point on the $i$-th orbit. Using the longitude $\theta_i$ and the inclination angle $\phi_i$, and the phase $\omega_{j}$ of the $j$-th satellite on this orbit, the $(x,y,z)$-coordinates of satellite $X_{i,j}$ are given as follows:
\begin{align}
	x&=r_{s}\sqrt{\cos^2(\omega_j)+\sin^2(\omega_j)\cos^2(\phi_i)}\cos(\hat{\theta}+\theta_{i}),\label{vecx}\\
	y&=r_{s}\sqrt{\cos^2(\omega_j)+\sin^2(\omega_j)\cos^2(\phi_i)}\sin(\hat{\theta}+\theta_{i}),\label{vecy}\\
	z&=r_{s}\sin(\omega_j)\sin(\phi_i)\label{vecz},
\end{align}
respectively, where $\hat{\theta}= \arctan2(\sin(\omega_j)\cos(\phi_i),\cos(\omega_{j})).$

Finally, the satellite Cox point process $\Psi$ is on the sphere $\bS=\{(x,y,z)\in\bR^3|x^2+y^2+z^2=r_s^2\} $ and it is denoted by
\begin{equation}
	\Psi = \sum_{i} \psi_i
\end{equation}
where $\cO$ denotes the orbit process and $\psi_{i}$ denotes the satellite point process on the $i$-th orbit. Since the orbits are isotropically distributed, namely invariant with respect to (w.r.t.) Earth rotation, and the satellite point processes $\psi_i$ on orbit $\{O_i\}_{i}$ are rotation invariant along their orbits, the final satellite point process $\Psi$ is rotation invariant \cite{10557592,10703111}. Fig. \ref{fig:1} shows the proposed Cox point process where orbital foundations are highlighted as black lines and satellites are exclusively located on those orbits. We use $\lambda=90$, $\mu=30$, $r_s=6921$ km, and $r_e=6371$ km.

\begin{remark}
	Since the satellites are all at the same altitude, the speeds of satellites on their orbits are the same. At time zero, the satellites on each orbit are given by a Poisson point process, and they rotate around the Earth on their orbits with the same speed. Therefore, by the displacement theorem \cite{baccelli2010stochastic}, the satellite point process at any given time has the same distribution as the satellite point process at time zero for each orbital plane; in other words, the satellite Poisson point process on each orbit is time-invariant. Moreover, since the satellite Cox point process is the collection of Poisson satellite point processes on orbits, the satellite Cox point process $\Psi$ is also time-invariant. The proof of this property can be found in \cite{10410220}.
\end{remark}

\subsection{Spatial Model for User Locations}
When satellite infrastructures are shared among network operators, any satellite from any operator may serve users as long as they fall within its coverage range. Depending on the application, the spatial distributions of network users associated with different operators can be heterogeneous. To generalize the analysis presented in this paper, we assume that network users are randomly distributed on the Earth’s surface, denoted by $\bS_{e} = \{(x, y, z) \in \bR^3 \mid x^2 + y^2 + z^2 = r_e^2\},$ where $r_e$ is the Earth’s radius. Specifically, user locations are modeled as a Poisson point process on \( \bS_e \) with intensity parameter $N_u$ for mathematical tractability.\footnote{In both binomial and Poisson point processes, points are uniformly distributed on a sphere. The binomial point process can be regarded as a special case of the Poisson point process when the number of points is fixed rather than Poisson distributed.} The user point process is isotropic.

Owing to the joint isotropy and independence of the satellite and user point processes, we can focus on a typical user located at the North pole, $u_0 = (0, 0, r_e)$. The network performance is then evaluated from the perspective of this typical user rather than from the perspective of each and every user individually. In fact, the performance observed by the typical user---such as the geometric coverage profile of the typical user---corresponds to the spatial average of the geometric coverage profiles observed from all users on Earth.

\subsection{Geometric Coverage Cells of Satellites}\label{section:range}
When the satellite infrastructure is shared, satellites provide downlink signals to users located within their communication ranges, or equivalently, to users inside their respective coverage cells. Each satellite cell—representing its geometric coverage—is defined as a symmetrical circular region on the Earth’s spherical surface produced by the satellite’s footprint,  motivated by practical observations~\cite{38821,38901}. In this paper, we assume that each satellite can provide downlink service to users up to a distance denoted by $\kappa$. The variable $\kappa$ depends on several implementation factors, including the transmit antenna aperture and the satellite transmit power.

To account for heterogeneous satellite networks with various link configurations, we assume that the random variable $\kappa $ is an i.i.d. uniform between two constants $c_m$ and $c_M$, where we put restriction based only on spherical physics as follows: $c_m = r_s - r_e$ and $c_M < \sqrt{r_s^2 - r_e^2}$. Let $\cC_{i,j}$ denote the geometric coverage cell of the satellite located at $X_{i,j}$. Based on the above description, we have
\begin{equation}
	\cC_{i,j} = \{(x, y, z) \in \bS_e | \|(x, y, z) - X_{i,j}\| \leq \kappa_{i,j}\},
\end{equation}
where $\kappa_{i,j}$ is the communication range of the satellite $X_{i,j}$. The total geometric coverage region---representing the union of all geometric coverage cells produced by the satellites---is given by
\begin{equation}
	\cC = \bigcup\limits_{X_{i,j}\in\Psi} \cC_{i,j}\label{eq:cell}.
\end{equation}
Note that the total coverage set accounts for the overlap of cells by taking the union of sets, instead of their arithmetic summation.
Since the cell centers are the radial projection of the satellite Cox point process $\Psi$ and the coverage cells are i.i.d. random sets, the total coverage set \(\cC\) is therefore a spherical Cox-Boolean model. See \cite{baccelli2010stochastic} for an example of a planar Boolean model centered on a two-dimensional Poisson point process.

\begin{definition}\label{D:1}
	The independent random variable of the communication range $\kappa_{i,j}$ of the satellite $X_{i,j}$ is a mark of the satellite point process. Then, the satellite point process providing geometric coverage to the typical user is characterized as a thinned Cox point process as follows:
	\begin{equation}\label{eq:13}
		\Psi^p = \sum_{X_{i,j}\in\Psi} \delta_{X_{i,j}}\epsilon_{X_{i,j}},
	\end{equation}
	where $\epsilon_{X_{i,j}} = \ind_{\kappa_{i,j}>\|X_{i,j}-u_0\|}$. Each satellite point is retained if its mark is greater than the distance from the typical user to the point.
\end{definition}

\subsection{Propagation Model for Received Signal and Interference}
To obtain the received signal power from a satellite to the typical user, we assume that the received signal power of the typical user at the distance $d$ from the satellite transmitter $X_{i,j}$ is given by $p_{i,j}G_{\text{tx};i,j}G_{\text{rx}}H_{i,j}(\lambda_{\text{carrier}}^2/16\pi^2)d^{-\alpha},$
where $p_{i,j}$ is the transmit power of the satellite, $G_{\text{tx};i,j}$ is the transmit antenna gain of the satellite, $G_{\text{rx}}$ is the receive antenna gain of the typical user, $\lambda_{\text{carrier}}$ is the wavelength of the carrier frequency, $H_{i,j}$ is an i.i.d. random variable for small-scale fading for signal fluctuation, and $\alpha$ is the path loss exponent ($\alpha\geq 2$).

\subsection{Performance Metric}\label{S:2E}
This paper analyzes the performance of satellite infrastructure sharing by first examining the underlying geometric structure of the coverage region and then relating it to communication-level performance. We begin with geometric metrics such as the connection probability and the connection number, which characterize whether a typical user is covered and how many satellites simultaneously provide coverage. These quantities capture the spatial overlap and redundancy induced by orbit-constrained deployments, and provide a direct description of the coverage structure experienced by the user.

Building on this geometric foundation, we then consider metrics that more directly reflect communication performance, including the association distance, the total received signal power, and a baseline SIR under a nearest-association rule. The association distance provides a geometric proxy for the strength of the desired signal, while the total received signal power and SIR account for the aggregate contribution of multiple covering satellites and the resulting interference. This progression from geometry to communication enables a systematic interpretation of how orbit-induced coverage structure influences link-level performance in heterogeneous satellite networks.

\subsubsection{Connection Probability}
First, we evaluate the size of the total geometric coverage relative to the entire Earth surface, by evaluating the connection probability. It is defined by the probability that the typical user is located within the geometric coverage cell of any satellite, at any given time. Thanks to the rotation invariance and independence, the connection probability of the typical user corresponds to the spatial average of the connection probability of all users in the network. The connection probability of the typical user is given by $p_c=\bP(\texttt{connection})=\bP(u_0\in\cC),$
where $\cC$ is the total geometric coverage region in Eq.~\eqref{eq:cell}. Note $\cC$ is the \emph{union} of the satellite cells, not the arithmetic summation of geometric coverage cells. The connection  probability takes a value between zero and one where one corresponds to the case when the union of all satellite cells is equivalent to the Earth sphere.

\subsubsection{Connection Number}
We consider the connection number which is the number of satellite geometric cells that contain the typical user at any given realization. This metric is linked to the number of simultaneous connections that the typical user can make in the proposed heterogeneous satellite network. The connection number is defined as a nonnegative discrete random variable $N_{c} = \sum_{X_{i,j}\in\Psi}\ind_{u_0\in \cC_{i,j}},$	where $\ind_{u_0\in \cC_{i,j}}$ is an indicator function and it takes the value of one if the event $\{u_0\in\cC_{i,j}\}$ occurs and zero otherwise.
	We evaluate the connection number by deriving its Laplace transform $\cL_{N_{c}}(t)=\bE[\exp(-{tN_{c}})]$. The average connection number $n_c$ is given by $  	n_{c}= \bE[N_{c}]$. Note the connection number and its expectation account for the number of potential handover candidates, and they indirectly demonstrate the robustness of the satellite infrastructure sharing in heterogeneous LEO satellite networks.

\subsubsection{Association Distance}
Each geometric coverage cell theoretically characterizes the boundary of the service area that exhibits non-negligible received signal strength on the Earth surface. Assuming that the typical user is associated with its nearest LEO satellite in coverage, the association distance is defined as the distance from the typical user to its nearest LEO satellite. By taking into account the geometric coverage cells, the association distance $Z$ is given by $\min_{X_{i,j}\in\Psi^p}\|X_{i,j}-u_0\|,$ which will be defined properly in Section~\ref{s:3-3} as a conditional expression.

\subsubsection{Received Signal and Total Received Powers}
The average received signal power is defined as the average power received from the associated satellite, which is the nearest to the typical user among all satellites offering coverage cells. On the other hand, the total received power is the arithmetic summation of all received signal powers from the satellites whose geometric coverage cells contain the typical user. The received signal power and total received power will be formally defined in Section~\ref{s:3-4}.
\subsubsection{SIR Distribution}
Based on the proposed geometric coverage cell model, we analyze the SIR of the typical user by linking the geometric coverage structure to communication performance. We assume that the typical user is associated with the nearest satellite whose coverage cell includes the user, while the remaining satellites that also provide coverage act as interferers. Under this nearest-association rule, we evaluate the complementary cumulative distribution function (CCDF) of the SIR experienced by the typical user. Since the sharing scenarios of interest involve multiple operators reusing the same spectrum---so that several satellites simultaneously cover the typical user---we adopt an interference-limited formulation and use the SIR as the link-level metric. The justification of this assumption and the regime in which it holds are made explicit in Section \ref{S:3E}.

Table \ref{Table:1} summarizes key variables and their descriptions.

\begin{table}\centering
	\caption{Notation and Description}\label{Table:1}
	\begin{tabular}{|c|c|}
		\hline
		Notation	&  Description \\
		\hline
		$\Xi$	&  Poisson point process on $\mathbb{T}$ \\
		\hline
		$O(\theta,\phi)$	& orbit with longitude $\theta$ and inclination $\phi$ \\
		\hline
		$\cO$	& Poisson orbit process in $\bR^3$\\
		\hline
		$r_s$   & orbit radius \\
		\hline
		$r_e$   & Earth radius \\
		\hline
		$\lambda$ & mean number of orbits \\
		\hline
		$\mu$ & mean number of satellites per orbit\\
		\hline
		$\Psi$	& satellite point process\\
		\hline
				$\psi$	& satellite point process on each orbit\\
		\hline
		$u_0$	& typical user location $(0,0,r_{e})$ \\
		\hline
		$\kappa_{i,j}$	& communication range of satellite $X_{i,j}$ \\
		\hline
					$c_m$	& minimum communication range\\
		\hline
			$c_M$	& maximum communication range\\
		\hline
		$\cC_{i,j}$	&  geometric coverage cell of $X_{i,j}$\\
		\hline
				$\cC$	&  coverage set of all satellites \\
		\hline
		$\epsilon_{X_{i,j}}$ & indicator for satellite retention\\
		\hline
		$\Psi^p$	& retained satellite point process\\
		\hline
		$\exp(x)$	& exponential function $x$, namely $e^{x}$ \\
		\hline
		${F}_{\kappa}$ & CDF of $\kappa$ \\
		\hline 
		
	\end{tabular}
\end{table}

\section{Performance Analysis}\label{section_performance_analysis}
We now analyze the performance of satellite infrastructure sharing in heterogeneous satellite networks by deriving the performance metrics in Section~\ref{S:2E}.
\subsection{Connection Probability}
\begin{theorem}\label{T:1}
The connection probability of the typical user is given by 
		\begin{equation}
		1- \exp\left({-\int_{0}^{\tilde{\varphi}}{\lambda\cos(\varphi)}\left(1-e^{-\frac{\mu}{\pi}\int_{0}^{\tilde{\omega}(\varphi)}\bar{F}_{\kappa}\left(\tilde{K}\left(\varphi,\eta\right)\right)\diff \eta}\right)\diff \varphi}\right),\nnb
	\end{equation}
	where 
	\begin{align}
		&\tilde{\varphi} = \arccos((r_e^2+r_s^2-c_M^2)/(2 r_e r_s)), \label{eq:tilde varphi}\\
		&\tilde{\omega}(\varphi) = \arcsin(\sqrt{1-\cos^2(\tilde{\varphi})\sec^2(\varphi)}),\label{eq:tilde omega}\\
		&\bar{F}_{\kappa}(x) = \bP(\kappa>x) \ \forall x\in\bR,\label{eq:F_bar}\\
		&\tilde{K}(\varphi,\eta) = \sqrt{r_s^2 - 2r_s r_e \cos(\varphi)\cos(\eta)+r_e^2},\label{eq:tildeK}
	\end{align}
	respectively.
\end{theorem}
\begin{IEEEproof}
	See Appendix \ref{A:1}.
\end{IEEEproof}

Theorem \ref{T:1} provides a closed-form expression for the connection probability, defined as the probability that the typical user is covered by {at least} one satellite's geometric coverage. The result takes the exponential form reflecting the void probability of the underlying point process and shows how connectivity emerges from the superposition of individual coverage cells.

It is important to note that the closed-form formula in Theorem \ref{T:1} provides an efficient tool for evaluating satellite infrastructure sharing across diverse scenarios, in contrast to the time-consuming system-level simulation approach. Obtaining even a single data point via Monte Carlo simulation requires generating a very large number of instances, which becomes particularly costly when system parameters such as $\lambda$ the number of orbits or $\mu$ the number of satellites per orbit are large. By contrast, the analytical formula enables rapid evaluation of connection probability across multiple parameter values. This tractability allows system designers and operators to quickly assess the gains of satellite infrastructure sharing.

\begin{figure}
	\centering
	\includegraphics[width=1\linewidth]{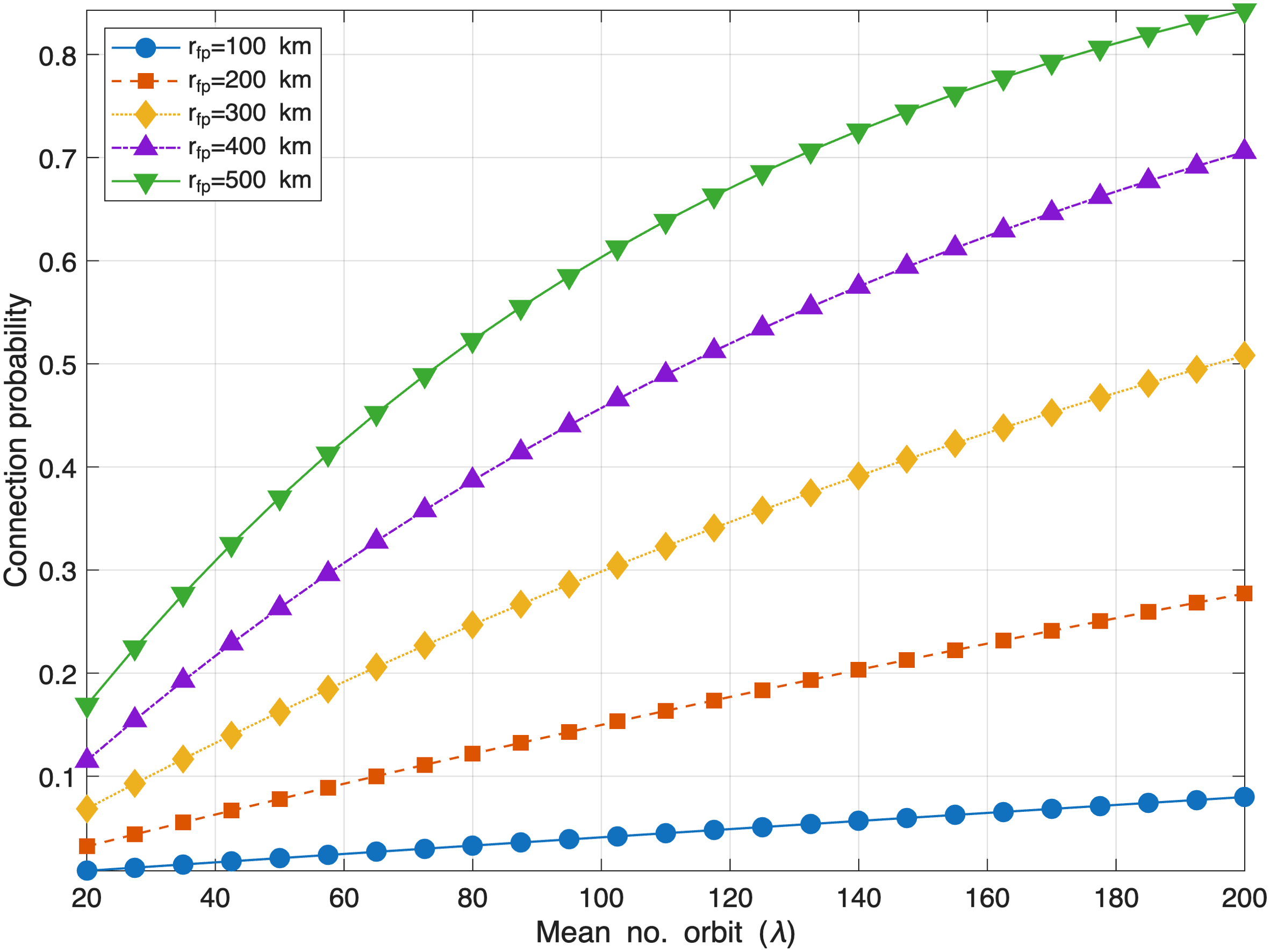}
	\caption{The connection probability of the typical user obtained by Theorem \ref{T:1} for different values of $\lambda$ and $c_M.$}
	\label{fig:2}
\end{figure}

Fig. \ref{fig:2} shows the connection probability for various values of $\lambda$ and $c_M$. Here, we set the mean number of satellites $\mu=15$ and the radius of orbit $r_s$ to be $6921$ km. In the figure, we use the footprint radii $r_{\text{fp}}$  to be  $100$ km, $200$ km, $300$ km, $400$ km, and $500$ km, which approximately give the maximum value of the range $c_M$ as $\sqrt{(r_s-r_e)^2 + (r_{\text{fp}})^2}$, resulting $c_M $ to be 559 km, 585 km, 626 km, 680 km, and 743 km. 

Suppose two independent operators each deploy a constellation with $\lambda = 20$ orbital planes and the same footprint parameter. From the figure, if $r_{\text{fp}} = 300$ km, or equivalently $c_M=626$ km, the connection probability for a typical user under a single operator is about $0.07$. If the two constellations are combined through infrastructure sharing, the effective number of orbital planes becomes approximately $\lambda = 40$. From the same curve, the connection probability increases to about $0.14.$ This corresponds to roughly a two-times improvement in connectivity for the typical user, highlighting how sharing aggregates coverage opportunities across independently deployed systems.

A similar trend is observed for larger footprints. For instance, when $r_{\text{fp}} = 500$ km, or equivalently $c_M=743$ km, a single operator at $\lambda = 20$ provides a connection probability of about $0.17$, whereas combining two such operators ($\lambda \approx 40$) increases the connection probability to roughly $0.32$, representing nearly a $1.9$-times gain. This example illustrates that, for various cell range, the infrastructure sharing effectively increases the spatial density of coverage cells, leading to a substantial improvement in connectivity, particularly in regimes where individual constellations are relatively sparse.

It is worth noting that the results are obtained under the condition that all satellite orbits are isotropically distributed. This highlights the practical benefit of the proposed satellite infrastructure sharing even when satellite orbits are randomly and individually deployed. In other words, satellite infrastructure sharing improves network connectivity significantly even when the orbital planes are neither coordinated nor jointly planned by the operators of heterogeneous satellite networks---a highly likely scenario in real-world deployments\footnote{The terms “uncoordinated” and “unplanned” denote a condition in which satellite orbits are established autonomously by stakeholders, subject only to collision avoidance.}.

\subsection{Connection Number}
The connection number is a nonnegative random variable that represents the number of satellites providing coverage to the typical user at any given time. We first analyze the Laplace transform of the connection number, which fully characterizes the probability mass function of the connection number random variable. We then analyze the expectation of it.
\begin{theorem}\label{T:2}
	The Laplace transform of the connection number is given by
	 \begin{align}
	&\exp\left(\!-\int_{0}^{\tilde{\varphi}}\!\!{\lambda\cos(\varphi)}\right.\nnb\\
	&\hspace{8mm}\left.\left(\!1-e^{-\frac{\mu}{\pi}\int_{0}^{\tilde{\omega}(\varphi)}(1-e^{-t})\bar{F}_\kappa(\tilde{K}(\varphi,\eta))\diff \eta}\right)\diff \varphi\right),\nnb
\end{align}
where $\tilde{\varphi}$, $\tilde{\omega}(\varphi)$, $\bar{F}_{\kappa}(x),$ and $\tilde{K}(\varphi,\eta)$ are given by Eqs. \eqref{eq:tilde varphi}, \eqref{eq:tilde omega}, \eqref{eq:F_bar}, and \eqref{eq:tildeK}, respectively.
\end{theorem}

\begin{IEEEproof}
	See Appendix \ref{A:2}.
\end{IEEEproof}
Theorem \ref{T:2} provides the Laplace transform (equivalently, the probability generating function) of the connection number, and therefore fully characterizes its distribution. Since the connection number is a nonnegative integer-valued random variable, its Laplace transform uniquely determines its probability mass function. As a result, the theorem gives complete statistical information on the number of satellites simultaneously covering the typical user, beyond only its first moment. 

To get a first-order behavior of the connection number, we analyze the expectation of the random variable below.

\begin{theorem}\label{T:3}
	The mean connection number of the typical user is given by
	\begin{equation}
		n_c=\frac{\lambda\mu}{\pi}\int_0^{\tilde{\varphi}}\int_{0}^{\tilde{\omega}(\varphi)}\cos(\varphi)\bar{F}_{\kappa}(\tilde{K}(\varphi,\eta))\diff \eta \diff \varphi,
	\end{equation}

where $\tilde{\varphi}$, $\tilde{\omega}(\varphi)$, $\bar{F}_{\kappa}(x)$, and $\tilde{K}(\varphi,\eta)$ are given by Eqs. \eqref{eq:tilde varphi}, \eqref{eq:tilde omega}, \eqref{eq:F_bar}, and \eqref{eq:tildeK}, respectively.
\end{theorem}

\begin{IEEEproof}
	See Appendix \ref{A:3}.
\end{IEEEproof}

Theorem \ref{T:3} shows that the mean connection number scales linearly with both the average number of orbits, $\lambda$, and the average number of satellites per orbit, $\mu$. This linear dependence provides a simple geometric interpretation: increasing either the number of orbital planes or the number of satellites placed on each orbit increases the expected number of satellites covering the typical user in direct proportion. In particular, the result makes explicit how orbit density and intra-orbit satellite density jointly determine the average coverage redundancy under infrastructure sharing.

\begin{figure}
	\centering
	\includegraphics[width=1\linewidth]{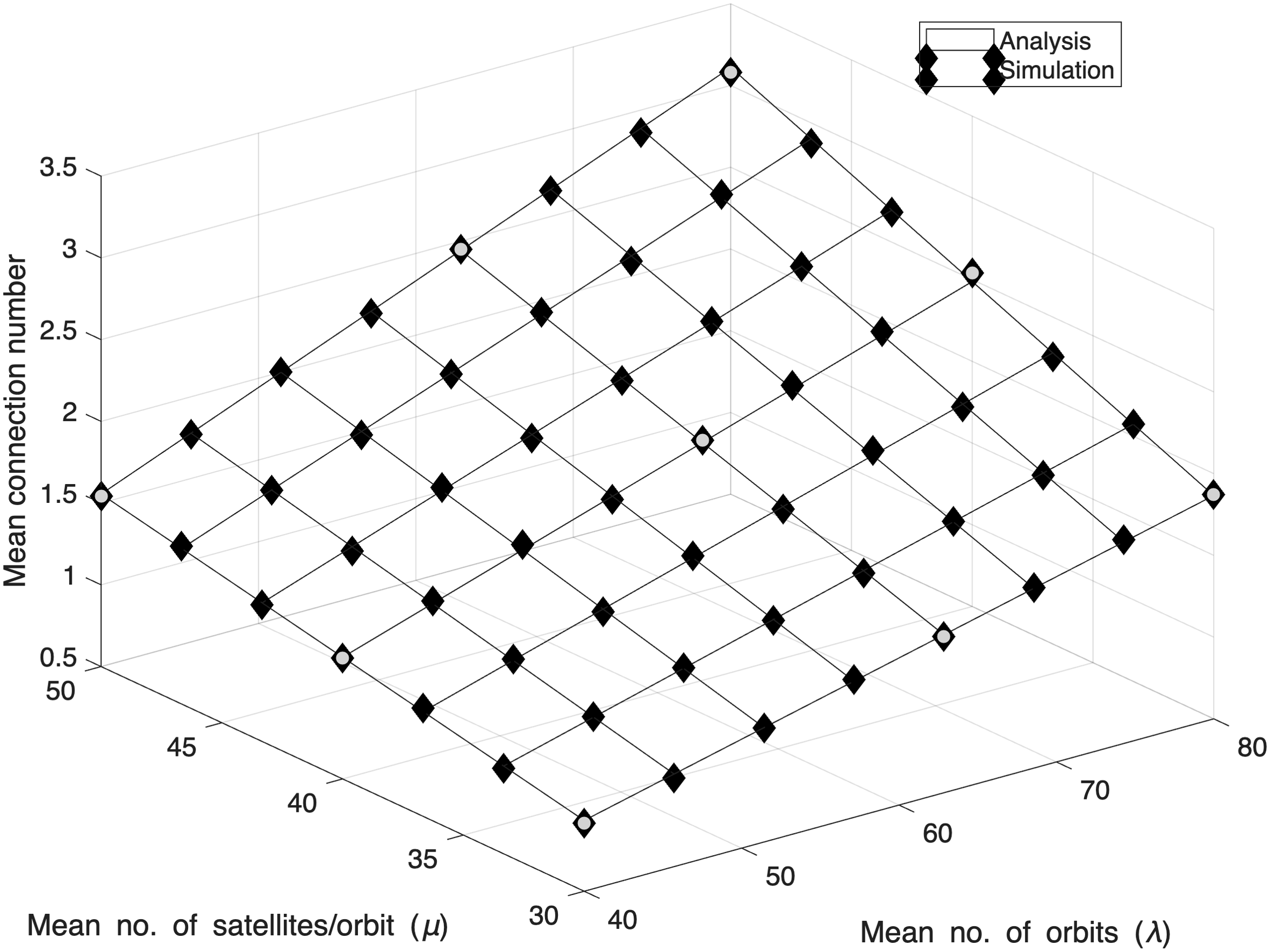}
	\caption{Illustration of the mean connection number obtained by simulation and analysis. We use $r_s=6921$ km and $c_M=770$ km. }
	\label{fig:3}
\end{figure}

Fig. \ref{fig:3} compares the mean connection number obtained from system-level simulations and from Theorem \ref{T:3}. Here, $c_M=770$ km corresponds to a footprint radius of approximately $r_{\text{fp}}\approx540$ km through $c_M \approx \sqrt{(r_s-r_e)^2+r_{\text{fp}}^2}$; this value is chosen as a representative intermediate coverage capability, lying between the footprint-limited ranges examined in Fig. \ref{fig:2} ($559$--$743$ km) and the larger ranges used in Figs. \ref{fig:4} and \ref{fig:5} ($916$--$1100$ km), so that the validation of Theorem \ref{T:3} covers the mid-range regime. The simulations are conducted using a Monte Carlo approach with $10^5$ independent realizations, where the connection number in each snapshot is averaged to obtain the empirical mean for given $\lambda$ and $\mu$. The close agreement between the analytical and simulation results across all parameter regimes confirms the accuracy of the derived expression.

Beyond validation, the figure reveals a clear structural trend: the mean connection number increases approximately linearly with both the mean number of orbital planes $\lambda$ and the mean number of satellites per orbit $\mu$. This behavior is consistent with the analytical result, indicating that the expected number of satellites covering the typical user scales proportionally with the overall satellite density, which is jointly determined by orbit-level and intra-orbit densities. Geometrically, increasing $\lambda $ introduces more orbital planes intersecting the Earth, while increasing $\mu$ densifies satellites along each orbit, and both effects contribute additively to coverage overlap.

In addition, the near-planar structure of the surface highlights that neither parameter dominates the other; rather, they play symmetric roles in determining connectivity. This suggests that, from a system design perspective, similar improvements in coverage redundancy can be achieved by increasing either the number of orbits or the number of satellites per orbit, providing flexibility in constellation design under practical constraints.

It is important to note that this result is derived under the isotropic Cox model, where orbital planes are assumed to be independently and isotropically distributed over the sphere. In practice, if satellite orbits exhibit structured or coordinated patterns—leading to non-isotropic configurations—the resulting infrastructure sharing gain may be reduced. Nevertheless, in scenarios where satellite systems are deployed independently by multiple operators without coordination, it is reasonable to assume that their orbital configurations are effectively independent, making the isotropic assumption a suitable first-order model. A more comprehensive assessment of this assumption is provided in Section \ref{S:5E}, where the connection number obtained from the proposed model is numerically compared with that of a down-sampled real-world constellation.

\begin{figure}
	\centering
	\includegraphics[width=1\linewidth]{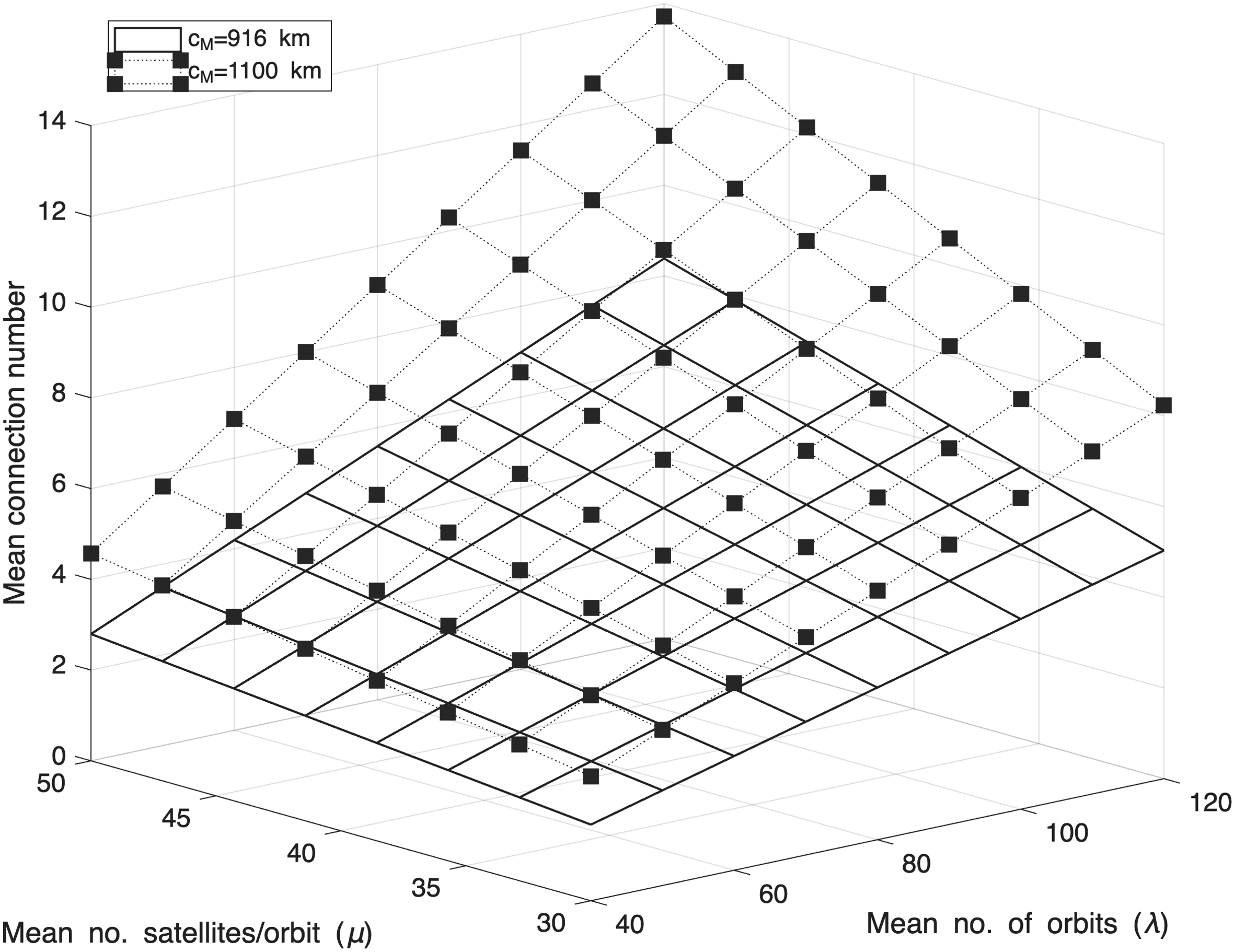}
	\caption{The mean connection number for two $c_M$ values. Specifically, $r_s=6921$ km, $c_m=550$ km, and $c_M=\{916 \text{ km}, 1100 \text{ km}\}$}
	\label{fig:4}
\end{figure}

Fig. \ref{fig:4} presents the mean connection number obtained using the formula in Theorem \ref{T:3}, with two different values of $c_M$. Here, $c_M$ denotes the maximum value that the communication range can take, while the communication range of each satellite is randomly distributed between $c_m=550$ km and the chosen $c_M$ value. The figure demonstrates that, for any $c_M$, the mean connection number increases almost linearly with both $\lambda$ and $\mu$. For example, the connection number increases from $1.7$ to $5.1$ (exactly $200\%$ improvement) when combining {three} uncoordinated satellite groups where each group is equipped with $\lambda=40$ and $\mu=30$. Similarly, for $c_M=1100$ km, the mean connection number reaches $8.23$ when $\lambda=120$ and $\mu=30$, which is approximately three times larger than the value of $2.74$ achieved by a single constellation with $\lambda=40, \mu=30$. In both cases, the performance improvement obtained from infrastructure sharing is identical to the arithmetic sum of the performances individually achieved by the three operators. For any $c_M$, this linear gain from infrastructure sharing is clearly substantiated by Theorem \ref{T:3}. Consistent with the observation made in Theorem \ref{T:1}, satellite infrastructure sharing in the heterogeneous networks increases the number of satellites within coverage, even when orbital planes are neither coordinated nor jointly planned.

\subsection{Distance to Association Satellite}\label{s:3-3}
Before we evaluate the distribution of the association distance, we now consider the communication ranges of satellites as independent marks on the satellite point process \cite{baccelli2010stochastic}.

The thinning of the point process is location-dependent thinning and yet the thinning of a point is independent of other points. Therefore, it is statistically independent and by first principle, conditionally on the orbit process, the thinned satellite point process on each orbit is a Poisson point process with modified intensity based on the retention probability~\cite{baccelli2010stochastic}. We now use the thinned point process $\Psi^{p}$ to derive the association distance. Let $Z$ be the distance from the typical user to its nearest satellite in coverage.

\begin{theorem}\label{T:4}
The CDF of the association distance $\bP(Z\leq z)$ is given as follows: for $z<c_m,$ $\bP(Z\leq z) =0$; for $c_m\leq z <c_M$, $\bP(Z\leq z)$ is given by
	\begin{equation*}
		1- \exp\left(\!-{\lambda}\!\int_{0}^{\hat{\varphi}(z)}\!\!\!\!\!\!\cos(\varphi)\left(1-e^{-\frac{\mu}{\pi}\int_{0}^{\hat{\omega}(z,\varphi)}\bar{F}_{\kappa}(\tilde{K} (\varphi,\eta))\diff \eta}\right)\diff \varphi\right),
	\end{equation*}
	where $\hat{\varphi}(z)$, $\hat{\omega} (z,\varphi)$, $\bar{F}_\kappa$, and $\tilde{K}(\varphi,\eta)$ are given by Eqs. \eqref{27_1}, \eqref{28}, \eqref{eq:F_bar}, and \eqref{eq:tildeK}, respectively; finally, for $c_M\leq z,$ we have $\bP(Z\leq z) = \bP(\texttt{connection})$ derived in Theorem \ref{T:1}.
\end{theorem}
\begin{IEEEproof}
	See Appendix \ref{A:4}.
\end{IEEEproof}

Let $\tilde{Z} $ denote the distance from the typical user to its nearest satellite, conditionally on the fact that the typical user is in the coverage cell of any satellite. The integration of the density function of $\tilde{Z}$ over $[c_m,c_M]$ gives one and hence, we refer to the variable $\tilde{Z}$ as the normalized association distance.

\begin{corollary}\label{C:1}
	For $c_m<z<c_M,$ the CDF of the normalized association distance $\tilde{Z}$,  namely $\bP(\tilde{Z}<z)$ is given by
	\begin{equation}
		\frac{1- e^{\left.-\int_{0}^{\hat{\varphi}(z)}\!{\lambda}\cos(\varphi)\left(1-e^{-\frac{\mu}{\pi}\int_{0}^{\hat{\omega}(z,\varphi)}\bar{F}_{\kappa}(\tilde{K} (\varphi,\eta))\diff \eta}\right)\diff \varphi\right.}}{1- e^{\left.{-\int_{0}^{\tilde{\varphi}}{\lambda\cos(\varphi)}\left(1-e^{-\frac{\mu}{\pi}\int_{0}^{\tilde{\omega}(\varphi)}\bar{F}_{\kappa}\left(\tilde{K}\left(\varphi,\eta\right)\right)\diff \eta}\right)\diff \varphi}\right.}},
	\end{equation}
where $\hat{\varphi}(z)$, $\hat{\omega} (z,\varphi)$, $\tilde{\varphi}$, $\tilde{\omega}(\varphi)$, $\bar{F}_\kappa$, and $\tilde{K}(\varphi,\eta)$ are given by Eqs.  \eqref{27_1}, \eqref{28}, \eqref{eq:tilde varphi}, \eqref{eq:tilde omega}, \eqref{eq:F_bar}, and \eqref{eq:tildeK}, respectively.
	\begin{IEEEproof}
		Based on the definition of $\tilde{Z},$ we have
		\begin{equation}
			\bP(\tilde{Z}<z)=\bP(Z<z|Z\neq \infty)=\frac{\bP(Z<z,Z\neq \infty)}{\bP(Z\neq \infty)}\nnb,
		\end{equation}
		where the numerator is given by Theorem \ref{T:4}.
	\end{IEEEproof}
\end{corollary}
Fig. \ref{fig:5} shows the CDF of the normalized association distance, namely $\bP(\tilde{Z}<z)$ obtained by the derived formula in Theorem \ref{T:4} and by system-level simulations. Comparing the numerical analysis with the simulation results asserts the accuracy of the formula we derived. In this figure, we use $\lambda=\{30,60,90\} $ to illustrate the behavior of the association distance in various scenarios: a single network operator, two network operators, and three network operators where each operator is equipped with $30$ orbital planes on average. As $\lambda$ grows or equivalently, as more and more operators are engaged in satellite sharing, the association distance decreases, illustrating the benefits in terms of desired link quality.
\begin{figure}
	\centering
	\includegraphics[width=1\linewidth]{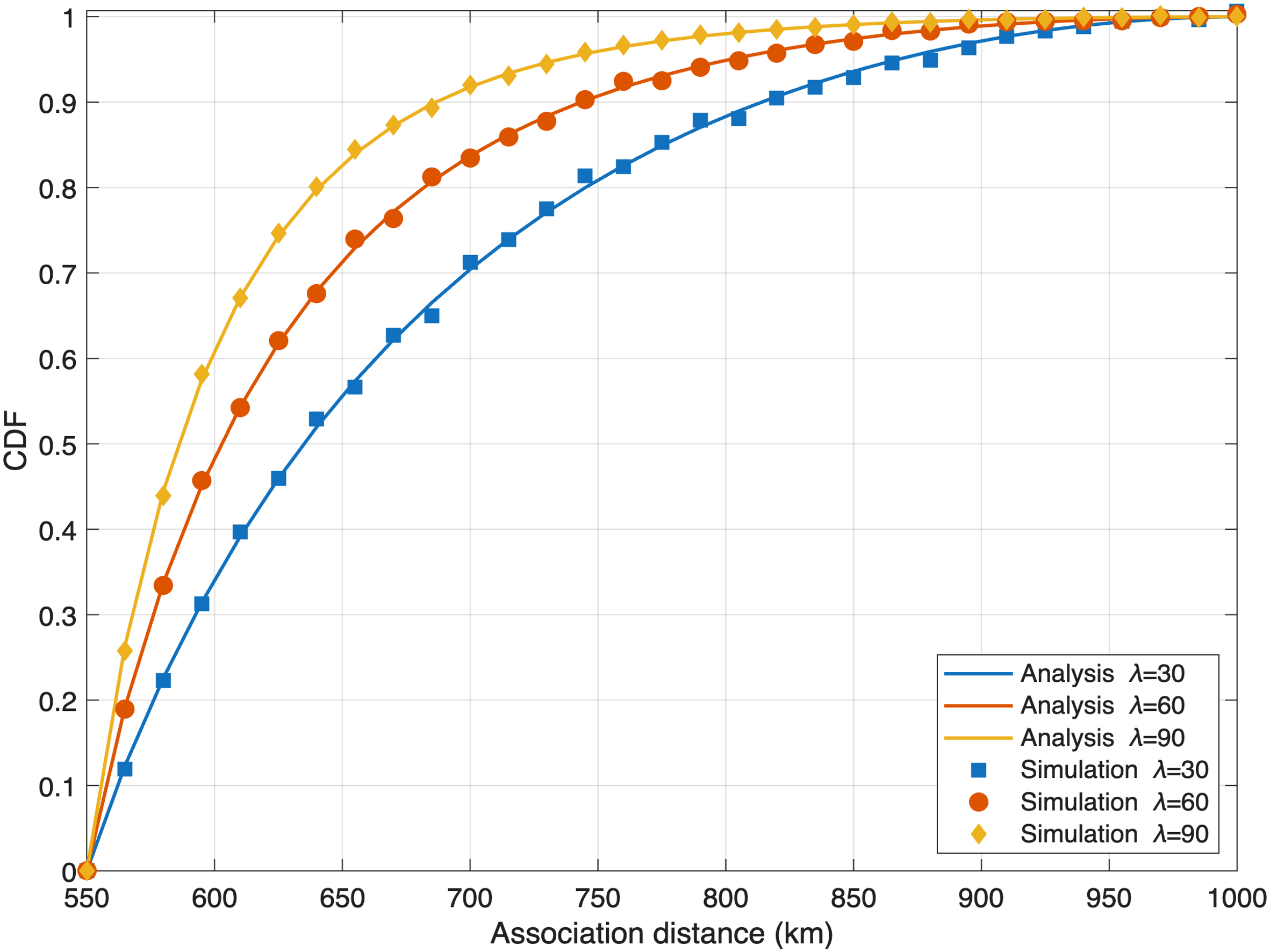}
	\caption{The normalized association distance distribution of the typical user when $r_s=6921$ km, $\mu=30$, and $c_M=1000$ km.}
	\label{fig:5}
\end{figure}

\subsection{Received Signal Power and Total Received Power}\label{s:3-4}
For the moment, let us assume the transmit power of satellite $X_{i,j}$ be $p_t,$ namely $p_{i,j}=p_t$ a constant. The total received power is defined as the total power received by the typical user from \emph{all} satellites within coverage. The total received power, denoted by $T$, is given by 
\begin{align}
	T &= \sum_{X_{i,j}\in\Psi^p} \frac{p_{t}G_{i,j}(\lambda_{\text{carrier}}^2/16\pi^2)H_{i,j}}{\|X_{i,j}-u_0\|^{\alpha}}\nnb\\
	&=\sum_{X_{i,j}\in\Psi^p} \frac{\bar{p}H_{i,j}}{\|X_{i,j}-u_0\|^{\alpha}},
\end{align}
where $\Psi^{p}$ is the thinned satellite point process, $G_{i,j}=G_{\text{tx};i,j}G_{\text{rx}}$ is the aggregate antenna gain between the satellite  $X_{i,j}$ and the typical user, $\lambda_{\text{carrier}}$ is the wavelength of the carrier frequency, $\alpha$ is the path loss exponent, and $H_{i,j}$ is an i.i.d. random variable associated with fading between $X_{i,j}$ and the typical user. We assume that $\bE[H]=1$ where the Laplace transform of $H$ is denoted by $\cL_{H}(s).$ To maintain the tractability of the analysis, we assume $\bar{p}=p_tG_{i,j}(\lambda_{\text{carrier}}^2/16\pi^2) $ for all $i$ and $j.$ Total received power $T$ outlines the amount of effective interference level expected by a random gateway on any location when network operators share their satellite infrastructures and they all use the same spectrum.

On the other hand, if the typical user is not within any of the coverage cells, there is no received signal and the received signal power is zero. Here, we use the normalized association distance $\tilde{Z}$, introduced in Corollary \ref{C:1}, to evaluate the average received signal power.
\begin{theorem}\label{T:5}
	First, the average total received power is
	\begin{align}
	\bE[T]=	\int_{0}^{\tilde{\varphi}}\frac{\lambda\mu \bar{p} \cos(\varphi)}{\pi}\int_{0}^{\tilde{\omega}(\varphi)}\frac{\bar{F}_{\kappa}(\tilde{K}(\varphi,\eta))}{\tilde{K}^{\alpha}(\varphi,\eta)}\diff \eta \diff \varphi.
	\end{align}
	On the other hand, the Laplace transform of the total received power, denoted by $\cL_{T}(s)$, is given by
	\begin{equation}
		\cL_{T}(s)=e^{-{\lambda}\int_{0}^{\tilde{\varphi}}\cos(\varphi)\left(1-e^{\left.-\frac{\mu}{\pi}\int_{0}^{\tilde{\omega}(\varphi)}\bar{F}_{\kappa}(\tilde{K}(\varphi,\eta))\left(1-\cL_{H}(\frac{s\bar{p}}{\tilde{K}^{\alpha}(\varphi,\eta)})\right)\diff \eta\right.} \right)\diff \varphi},
	\end{equation}
	where $\tilde{\varphi}$, $\tilde{\omega}(\varphi),$ $\bar{F}_{\kappa}, $ and $\tilde{K}(\varphi,\eta)$ are given by Eqs. \eqref{eq:tilde varphi}, \eqref{eq:tilde omega}, \eqref{eq:F_bar}, and \eqref{eq:tildeK}, respectively.

	The average received signal power is given by
	\begin{equation}
		\bar{p} \left.\int_{0}^{\infty}v^{-\alpha}\frac{\partial \bP(\tilde{Z}\leq v)}{\partial v}\diff v \right.,
	\end{equation}
	where $\bP(\tilde{Z}\leq z)$ is given by Corollary \ref{C:1}.
\end{theorem}

\begin{IEEEproof}
	See Appendix \ref{A:5}.
\end{IEEEproof}
\begin{figure}
	\centering
	\includegraphics[width=1\linewidth]{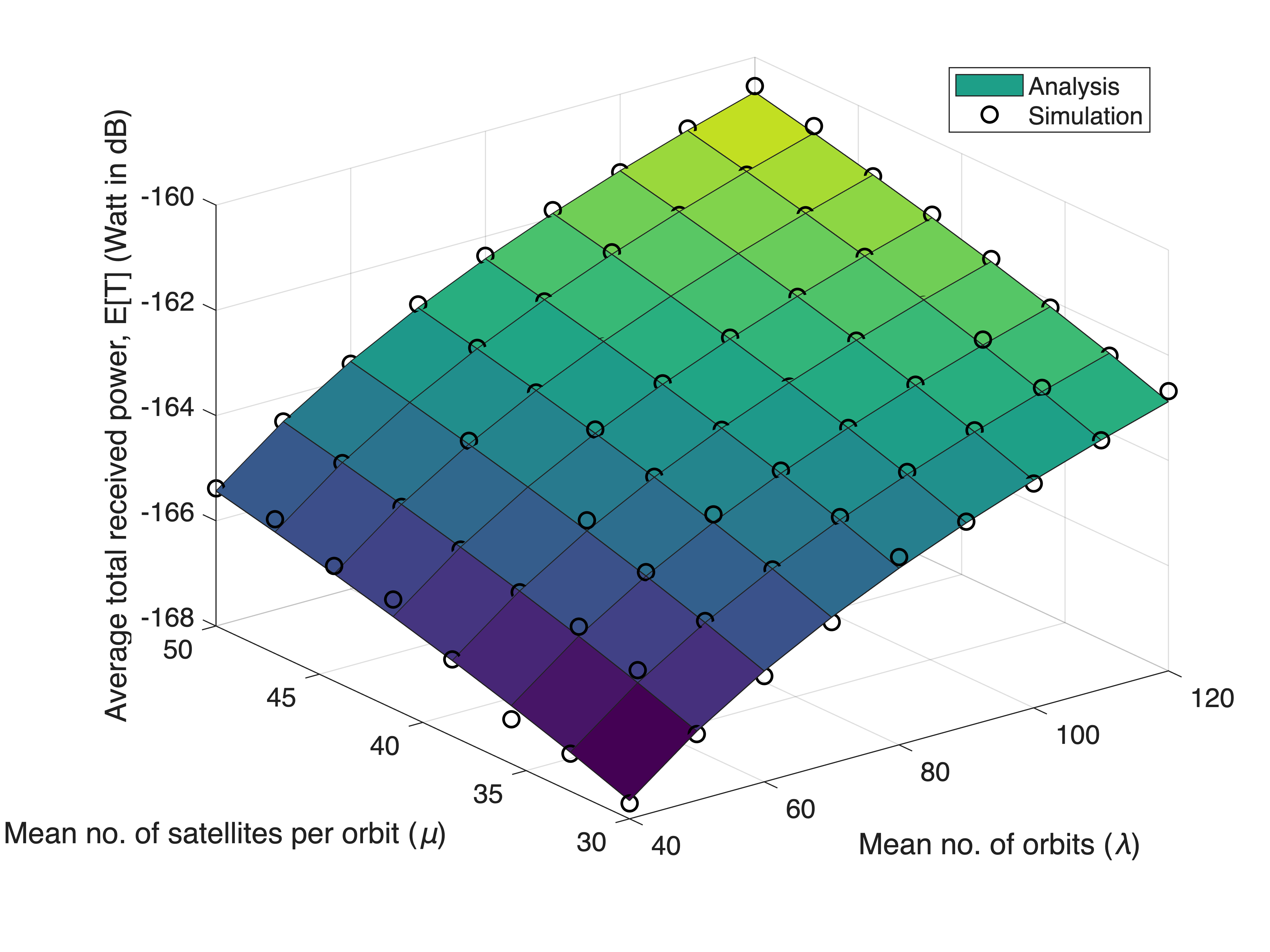}
	\caption{The average total received signal power $\bE[T]$.}
	\label{fig:6}
\end{figure}
Fig. \ref{fig:6} compares the average total received power at the typical user obtained by system-level simulation and by Theorem \ref{T:5}. It shows that the simulation results confirm the analytically derived formula. For the analysis, we adopted the parameters from \cite{38821}: $c_M=916$ km, $p_t=100$ mW, $G_{i,j}=2,$ $\lambda_{\text{carrier}}=0.06$ m, and $r_s=6921$ km.

\subsection{Numerical Result: Communication Performance}\label{S:3E}
We assume that the typical user at $u_0$ is associated with the nearest satellite among the ones whose cells contain the typical user. Furthermore, let us assume for the moment $\bar{p}=p_tG_{i,j}(\lambda_{\text{carrier}}^2/16\pi^2) $ to emphasize the impact of geometric orbital structure to the SIR. Then, using the thinned satellite point process $\Psi^p$ in Eq. \eqref{eq:13}, the SIR of the typical user is given by
\begin{equation}
	\SIR= \frac{\bar{p}H\|X_\star-u_0\|^{-\alpha}}{\sum_{X_{i,j}\in\Psi^p\setminus\{X_\star\}}\bar{p}H_{i,j}\|X_{i,j}-u_0\|^{-\alpha}},
\end{equation}
where the association satellite is given by
\begin{equation}
	X_\star = \argmin_{X_{i,j}\in\Psi^p}\|X_{i,j}-u_0\|.
\end{equation}
Then, using threshold $\tau,$ the SIR coverage is given as follows:
\begin{align}
&\bP(\SIR>\tau)\nnb\\
&\hspace{4mm}=\bP\left(\frac{\bar{p}H\|X_\star-u_0\|^{-\alpha}}{\sum_{X_{i,j}\in\Psi^p\setminus\{X_\star\}}\bar{p}H_{i,j}\|X_{i,j}-u_0\|^{-\alpha}}>\tau\right).\nnb
\end{align}

The SIR-based formulation above presumes an interference-limited operation where multiple operators reuse the same spectrum, and the interferers are precisely the co-covering satellites, namely the satellites of $\Psi^p$ other than the serving satellite. Whenever the connection number is two or greater, these interferers are located at distances comparable to the association distance---both are bounded by the same maximum communication range $c_M$---so the aggregate interference received over the shared spectrum is of the same order as the desired signal power and dominates the noise floor of typical LEO downlink budgets. This is exactly the regime examined in Figs. \ref{fig:7} and \ref{fig:8}: with $\mu=20$ and $c_M\in[650,900]$ km, the typical user---conditionally on being in geometric coverage---is covered by two or more satellites with substantial probability. For instance, for $\lambda=60$ and $c_M=900$ km, system-level simulations indicate that the conditional mean connection number is approximately $2.1$ and that the conditional probability of two or more covering satellites is approximately $0.6$; both quantities grow with $\lambda$ and $c_M$ (cf. Theorem \ref{T:3} and Fig. \ref{fig:3}). In the scenario in which exactly one satellite covers the user, there is no co-channel interferer, and such realizations are naturally counted as covered in the unconditional CCDF defined below. The SIR results reported below therefore apply to, and should be interpreted within, this interference-limited multi-operator sharing regime. The exception is sparse configurations in which the typical user is covered by at most one satellite---e.g., the single-operator deployments with connection probability below $0.1$ in Fig. \ref{fig:2}---where the link is instead noise-limited and the SIR is not the appropriate link-level indicator; there, the geometric metrics (connection probability, connection number, and association distance) and the received-power characterization in Theorem \ref{T:5} directly apply, and a noise term can be incorporated by evaluating $\bar{p}H\|X_\star-u_0\|^{-\alpha}/(\sigma^2+I)$ with the same simulation framework, where $\sigma^2$ is the noise power and the aggregate interference $I$ is equal to zero.

To complement the geometric coverage analysis conducted above, we evaluate the SIR distribution observed at the typical user under the proposed Cox-based orbit-structured deployment. Under exponential fading, namely $H$ follows an exponential random variable with mean $1$, the SIR is computed for each network realization, and the empirical CCDF, $\bP(\SIR>\tau)$, is obtained over multiple realizations. This framework allows us to directly relate the coverage overlap induced by orbital geometry to communication-level performance. In Figs. \ref{fig:7} and \ref{fig:8}, we use $r_s=6921$ km, $\mu=20$, and $\alpha=2.$ It is important to note that the SIR distributions presented in Figs. \ref{fig:7} and \ref{fig:8} are unconditional, meaning that they are computed over all network realizations, including those in which the typical user is not covered by any satellite. In such cases, the SIR is set to zero. As a result, the reported CCDF, $\bP(\SIR>\tau)$, jointly reflects both the probability of being in geometric coverage and the quality of the link when geometric coverage is available. We emphasize that this unconditional formulation is particularly relevant in infrastructure-sharing scenarios, where both coverage availability and interference jointly determine the effective performance experienced by the user.

Fig. \ref{fig:7} illustrates the impact of the mean number of orbital planes $\lambda$ on the SIR distribution. As $\lambda$ increases, the SIR distribution improves in the moderate SIR regime, as reflected by the upward shift of the CCDF surface. This behavior is due to the increased spatial diversity of orbital planes, which raises the likelihood that the typical user is covered by a closer satellite. However, this improvement gradually saturates as $\lambda$ becomes large. The reason is that increasing $\lambda$ also increases the number of satellites whose coverage regions overlap at the user location, thereby introducing additional interference. As a result, while increasing the number of orbital planes improves coverage reliability, its impact on high-SIR performance is limited.

Fig. \ref{fig:8} shows the impact of the maximum communication distance $c_M$ on the SIR distribution. In contrast to the behavior observed with $\lambda$, increasing $c_M$ leads to a degradation of the SIR distribution across all SIR thresholds. While a larger $c_M$ expands the coverage region and increases the probability that the user is served, it simultaneously enlarges the set of satellites whose coverage cells include the user. This results in a significant increase in interference, which outweighs the benefit of improved coverage. Consequently, the SIR decreases as $c_M$ increases.

\begin{figure}
	\centering
	\includegraphics[width=1\linewidth]{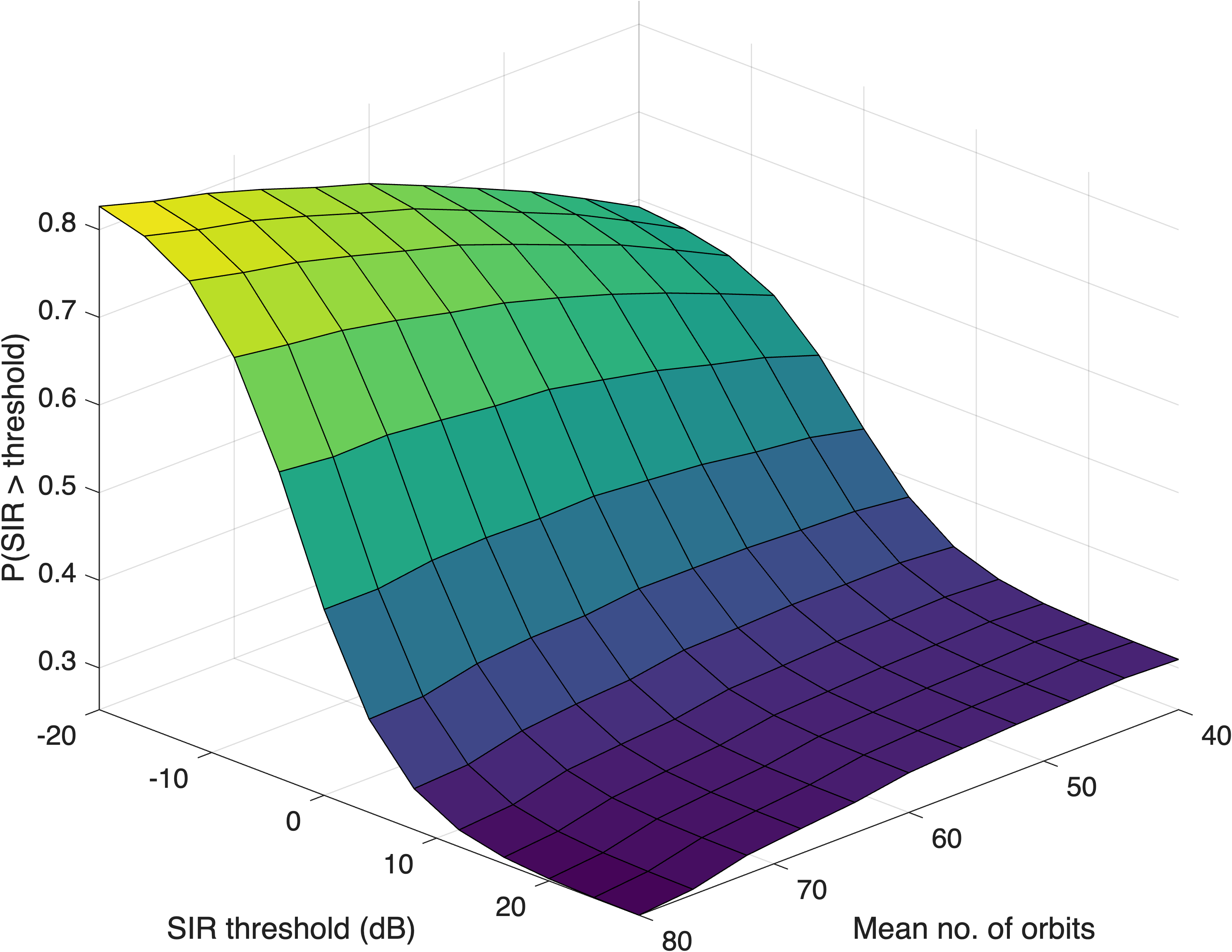}
	\caption{The SIR coverage probability by varying $\lambda$ the mean number of orbits. $c_M$ is set to be $900$ km.}
	\label{fig:7}
\end{figure}
\begin{figure}
	\centering
	\includegraphics[width=1\linewidth]{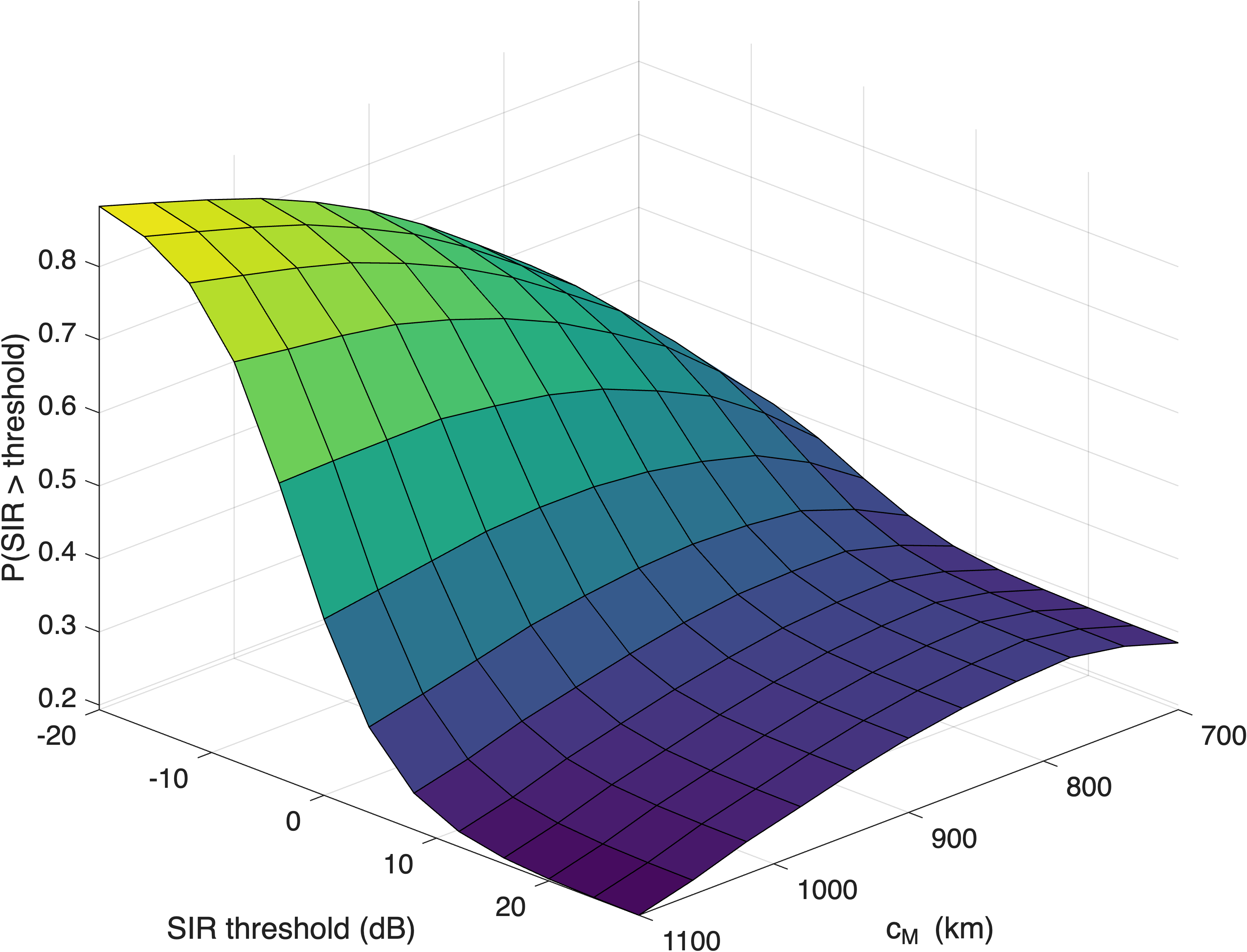}
	\caption{The SIR coverage probability by varying $c_M$ the maximum communication distance. $\lambda$ is set to be $60$.}
	\label{fig:8}
\end{figure}
This behavior is counter-intuitive from a coverage-centric perspective, where increasing the communication range is typically associated with improved performance. However, under infrastructure sharing, a larger $c_M$ significantly increases the degree of coverage overlap, leading to a substantial rise in interference. Since the desired signal is provided by a single serving satellite while interference accumulates from multiple satellites, the overall SIR is reduced. This highlights that improving geometric coverage does not necessarily translate into improved communication performance in orbit-structured networks.

Taken together, these results reveal a fundamental tradeoff between coverage and interference in orbit-structured satellite networks. Increasing either the number of orbital planes or the communication range improves geometric connectivity, but also increases the level of interference due to coverage overlap. As a result, key performance metrics such as the SIR are governed by the interaction between orbital geometry and coverage structure, rather than by node density or communication range alone. This observation underscores the importance of orbit-aware models, as spatially independent models such as binomial or Poisson models would not capture the coupling between coverage expansion and interference growth observed here.

\section{Further Discussion}\label{section:discussion}
\subsection{Multi-Altitude Satellite Networks}
In this section, instead of considering all orbits at the same altitude, let us assume for the moment that there are $M$ different values for satellite altitudes. Specifically, we model the locations of satellites as a summation of independent Cox point processes $\sum_{k}\Psi_k$ where the Cox point process $\Psi_k$ is equipped with $\lambda_k$ mean number of orbits, $\mu_{k}$ mean number of satellites per orbit, and $r_{s,k}$ the radius of orbit for $k=1,\ldots,M$. Since $M$ independent Cox point processes are combined, the final satellite point process $\tilde{\Psi}$ is given by $\tilde{\Psi}= \sum_{k=1}^{M}\Psi_k = \sum_{k=1}^{M}\sum_{X_{i,j}\in\Psi_k}\delta_{X_{i,j}},$
where $\Psi_k$ is the satellite Cox point process with parameters $\lambda_k,\mu_k$ at radius $r_{s,k}$, for $k=1,\ldots,M$.

To maintain the tractability of the analysis, the communication range of each satellite must be greater than its altitude, exactly as in the single-altitude model of Section \ref{section:range}. Hence, for a given $k$, the communication range of the satellite $X_{i,j}\in\Psi_k$, denoted by $\kappa_{i,j,k}$, is assumed to be an i.i.d. uniform random variable between $c_{m,k}$ and $c_{M,k}$, where the minimum value is set to $c_{m,k}=r_{s,k}-r_e$ so that every realization of $\kappa_{i,j,k}$ exceeds the altitude of the satellite. We let $F_{\kappa,k}$ denote the CDF of the communication range of the satellites of $\Psi_k$, and $\bar{F}_{\kappa,k}=1-F_{\kappa,k}$ its CCDF.

It is important to note that satellites at multiple altitudes are now fully modeled as the superposition of independent Cox point processes. Hence, we use the machinery in Section \ref{section_performance_analysis} to examine the performance of the proposed satellite infrastructure sharing under the practical assumption that satellites are now at various altitudes.

Below, we present the derivation of the connection probability of the typical user. It is worth noting that the single-altitude analysis extends to the multi-altitude case in a straightforward manner within the proposed analytical framework.
\begin{proposition}\label{P:1}
	With the superposition model, the connection probability of the typical user is given by
	\begin{align}
		\bP(\texttt{connection})&=1- \prod_{k=1}^M\bP(\texttt{outage for $\Psi_k$}),\nnb
	\end{align}
where the outage probability of the Cox point process $\Psi_k$ is
\begin{align}
	&\exp\left(-\int_{0}^{\tilde{\varphi}_{k}}{\lambda_{k}\cos(\varphi)}\right.\nnb\\
	&\hspace{10mm}\left.\left(1-e^{-\frac{\mu_{k}}{\pi}\int_{0}^{\tilde{\omega}_{k}(\varphi)}\bar{F}_{\kappa,k}(\tilde{K}_{k}(\varphi,\eta))\diff \eta}\right)\diff \varphi\right),\nnb
\end{align}
where for $k=1,\ldots,M$, we have
	\begin{align}
	&\tilde{\varphi}_{k} = \arccos((r_e^2+r_{s,k}^2-c_{M,k}^2)/(2 r_e r_{s,k})), \label{eq:tilde varphi_k}\\
	&\tilde{\omega}_{k}(\varphi) = \arcsin(\sqrt{1-\cos^2(\tilde{\varphi}_{k})\sec^2(\varphi)}),\label{eq:tilde omega_k}\\
	&\tilde{K}_{k}(\varphi,\eta) = \sqrt{r_{s,k}^2 - 2r_{s,k} r_e \cos(\varphi)\cos(\eta)+r_e^2}.\label{eq:tildeK_k}
\end{align}
\end{proposition}
\begin{IEEEproof}
	Since $\tilde{\Psi}$ is the superposition of independent Cox point processes, we apply the machinery employed in Theorem \ref{T:1} to obtain the final result.
\end{IEEEproof}
\begin{figure}
	\centering
	\includegraphics[width=.97\linewidth]{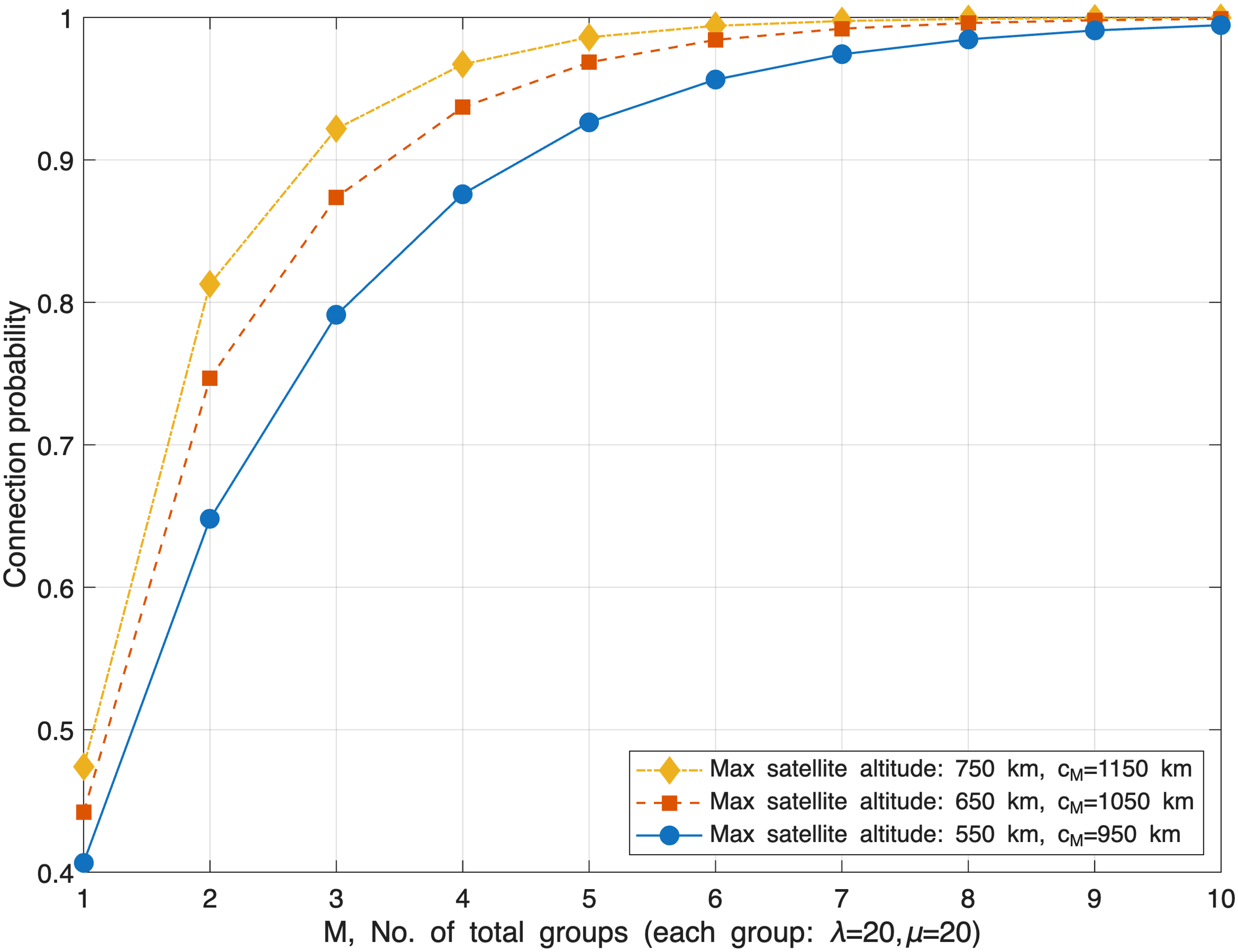}
	\caption{The connection probability of the typical user with multi-altitude satellite network operator groups. The $x$-axis is $M$.}
	\label{fig:9}
\end{figure}

Fig. \ref{fig:9} shows the connection probability numerically evaluated through Proposition \ref{P:1}. The $x$-axis is the number of satellite groups, or the variable $M$ in Proposition \ref{P:1}. Hence, the number of groups sharing the satellite infrastructure is indicated by a given $x$ value. Each group has $\lambda_k=20$ and $\mu_k=20$ for all $k=1,\ldots,M$. The curve with blue circles is a benchmark graph when all satellite groups have the same satellite altitude. On the other hand, the curve with red squares describes the connection probability when the satellite altitudes of different groups are uniformly distributed between $550$ km and $650$ km; similarly, the curve with yellow diamonds indicates the connection probability when the satellite altitudes are uniformly distributed between $550$ km and $750$ km. From the figure, we learn that the proposed satellite infrastructure sharing operates very well as intended, even for various satellite altitudes.  Along with the previous results with satellites being at the same altitude, this substantiates the applicability and feasibility of the proposed satellite sharing architecture in practical real-world situations. For the rest of the performance metrics, we can apply the techniques used in the proofs of Theorems \ref{T:2}–\ref{T:5} to assess the feasibility of satellite infrastructure sharing under a multi-altitude satellite constellation. A detailed analysis of these performance metrics is left for future work. We next validate the geometric accuracy of the proposed model by comparing it with a realistic constellation.
\begin{figure}
	\centering
	\includegraphics[width=1\linewidth]{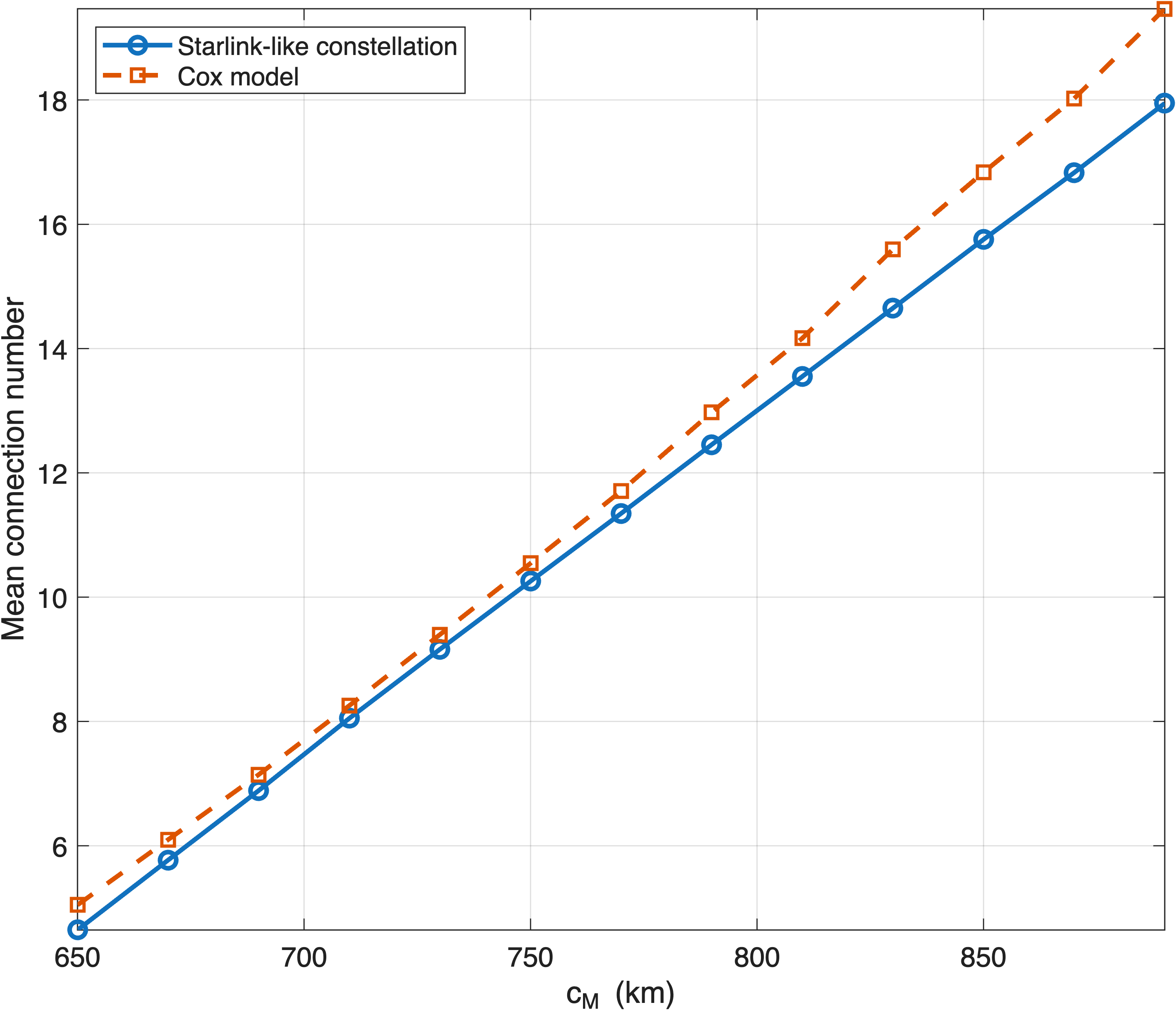}
	\caption{Mean connection number of the typical user as a function of $c_M$, comparing the Cox model and a Starlink-like constellation at latitude $30^\circ$. The Cox parameters are calibrated based on visibility statistics.}
	\label{fig:10}
\end{figure}
\subsection{Connection Number: Cox vs. Starlink}\label{S:5E}
To assess the practical relevance of the proposed Cox-based framework, we compare the connection number predicted by the Cox model with that obtained from a Starlink-like multi-shell constellation. The connection number, defined as the number of satellites whose geometric coverage cells include the typical user, is adopted as the primary comparison metric, motivated by its direct geometric interpretation. It captures the degree of coverage overlap induced by orbit-constrained deployments, and hence reflects both connectivity redundancy and the potential for handover and interference under infrastructure sharing. Unlike conventional metrics such as nearest distance or SIR, the connection number isolates the underlying geometric structure without being confounded by fading or link adaptation, making it particularly suitable for evaluating the validity of spatial models.

The down-sampled realistic constellation is constructed explicitly as follows. We consider a Starlink-like configuration consisting of three shells at a common altitude of $550$ km, i.e., $r_s=6921$ km, with inclinations of $43^\circ$, $53^\circ$, and $33^\circ$, respectively, motivated by the inclination structure of the Starlink Gen2 deployment authorized in \cite{FCCStarlink}. Each shell follows the Walker constellation structure \cite{walker1984satellite,11159552}: it comprises $25$ orbital planes whose ascending nodes are evenly spaced, and each orbital plane carries $100$ evenly phased satellites. The per-plane population is obtained by down-sampling a fully populated reference plane of $125$ satellites; the term ``down-sampled'' refers to this reduction, which keeps the total constellation size ($7500$ satellites) at a scale comparable to the deployments analyzed in this paper while preserving the Walker-type orbital structure. In each Monte Carlo realization, an independent random rotation of the ascending nodes is applied to each shell and an independent random phase offset is applied to each orbital plane, emulating the motion of the constellation relative to the typical user, located at latitude $30^\circ$. A satellite is declared visible when its distance to the typical user does not exceed $\sqrt{r_s^2-r_e^2}$, i.e., when it is above the local horizon. As in the rest of the paper, each satellite is assigned an i.i.d. communication range $\kappa$, uniformly distributed on $[c_m,c_M]$ with $c_m=550$ km.

The comparison is carried out by calibrating the Cox model parameters using visibility statistics obtained from the Starlink-like constellation over $10^4$ independent realizations. Specifically, the mean number of orbital planes $\lambda$ in the Cox model is set to match the average number of visible orbits observed from the typical user, while the mean number of satellites $\mu$ is set to match the average number of visible satellites per realization. The mean connection number of the two models is then estimated over $3\times10^{4}$ independent realizations for each value of $c_M\in[650,900]$ km. This calibration ensures that both models are aligned in terms of first-order geometric exposure, while preserving their intrinsic structural differences. We note that the Cox model is isotropic by construction, whereas the Starlink constellation exhibits structured and non-isotropic orbital configurations. As such, the proposed calibration serves as a first-order approximation that enables a meaningful comparison between the two models, despite their fundamentally different geometric characteristics. While alternative calibration strategies could be considered, the adopted approach provides a simple and transparent baseline for evaluating how well the Cox model captures the geometric behavior of realistic deployments.

Finally, we emphasize that the comparison focuses on the Cox and Starlink models, rather than including spatially independent models such as binomial or Poisson point processes. This is because such models do not incorporate the orbital foundation that governs satellite placement and clustering, and therefore cannot capture the plane-induced spatial correlations that are central to the problem considered here. The objective of this section is thus not to benchmark all possible spatial models, but to examine whether the orbit-aware Cox model can provide a meaningful approximation of realistic satellite constellations in terms of their geometric coverage structure.

Fig. \ref{fig:10} shows that the mean connection number increases approximately linearly with $c_M$ for both the Cox model and the Starlink-like constellation. This is expected, as a larger $c_M$ expands the coverage region and increases the number of satellites whose coverage cells include the typical user. The two curves exhibit close agreement across the entire range of $c_M$, indicating that the Cox model, when calibrated using visibility statistics, captures the overall trend and scale of the connection number observed in the Starlink-like constellation. The similar growth behavior suggests that the average spatial exposure experienced by the typical user is reasonably approximated by the Cox model. A small but consistent gap is observed, with the Cox model slightly overestimating the connection number. This can be attributed to its isotropic construction, which distributes orbital planes uniformly over all orientations. In contrast, the Starlink-like constellation has fixed inclinations, leading to anisotropic coverage patterns and slightly reduced spatial diversity at the user location.

Overall, the results indicate that the connection number is primarily governed by orbit-induced visibility and coverage geometry, and that the Cox model provides a reasonable first-order approximation of this behavior despite its simplified construction.

\subsection{Extension to Communication Range Distribution}
\begin{figure}
	\centering
	\includegraphics[width=1\linewidth]{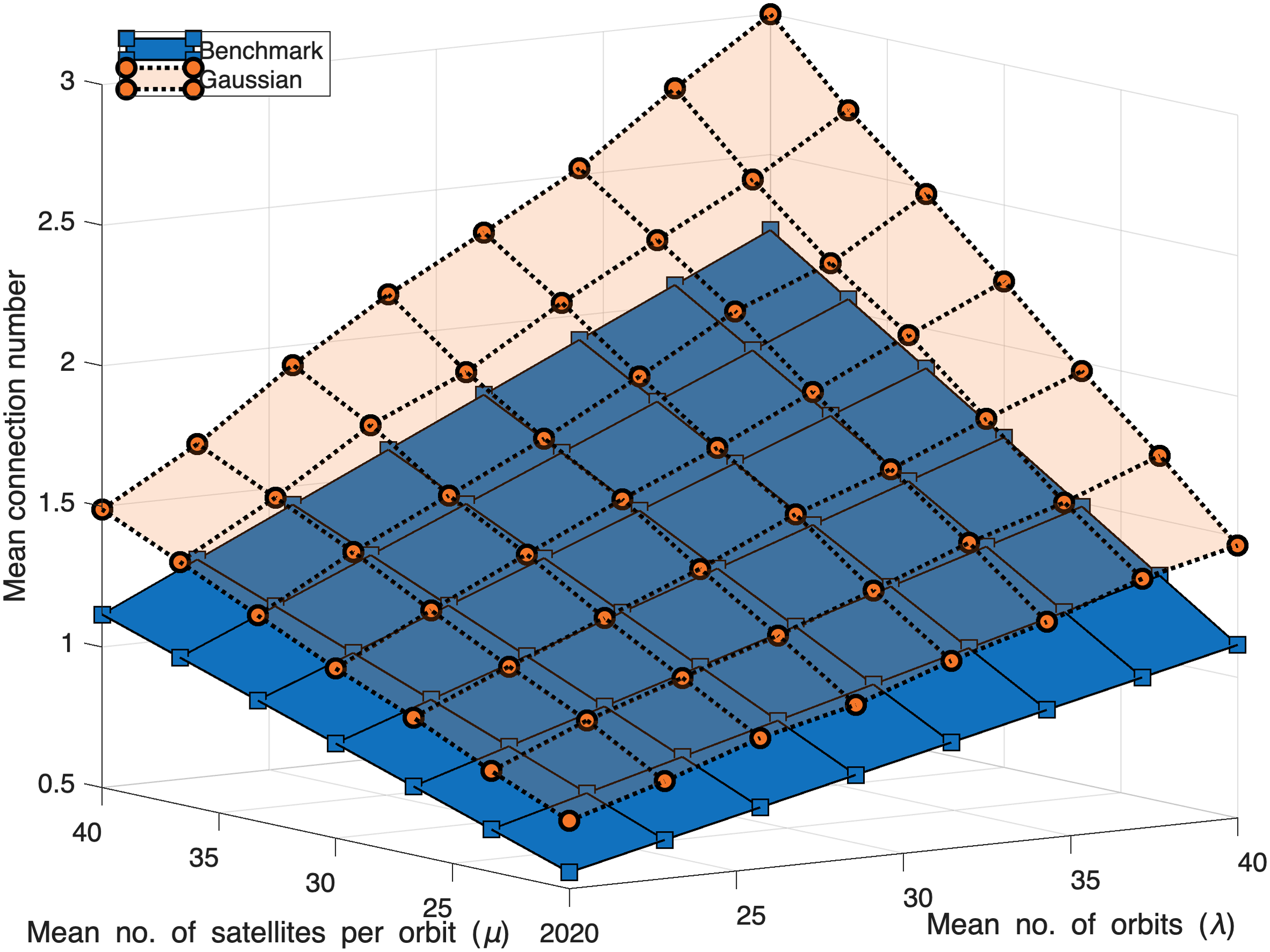}
	\caption{The mean connection number of the typical user with communication ranges for uniform and truncated Gaussian distribution, respectively.}
	\label{fig:11}
\end{figure}

In Section \ref{section_performance_analysis}, we derived the final results as general expressions of $F_{\kappa}(\cdot)$, the CDF of the communication range. Exploiting this expression, we further explore another distribution for the communication range so that the performance of the proposed network architecture can be evaluated under more diverse scenarios. For instance, in this section, we investigate a truncated Gaussian on $[c_m,c_M]$ with mean $\zeta$ and variance $\nu^2$. Specifically, a communication range sample value $x$ is drawn from a Gaussian distribution $\mathcal{N}(\zeta,\nu^2)$, and if $x \leq c_m$, it is rounded up to $c_m$, whereas if $x \geq c_M$, it is rounded down to $c_M$.

Fig. \ref{fig:11} shows the mean connection number obtained by the analytical expression in Theorem \ref{T:3} and by system-level simulations. The colored surface shows the benchmark scenario where the communication range follows an i.i.d. uniform distribution on $[c_m,c_M]$; in contrast, the orange surface with circles illustrates the scenario for Gaussian communication distance with mean $\zeta={c_m}/{3}+{2c_M}/{3}$, standard deviation $\nu=\sqrt{(c_M-c_m)^2/12}$, truncated between $c_m=550$ km and $c_M=916$ km. Here, we consider $r_s=6921$ km, $c_m=550$ km, and $c_M=916$ km.

\section{Conclusion}

This paper presents an analytical framework to evaluate the feasibility and performance of satellite infrastructure sharing among multiple LEO or MEO satellite operators. By modeling the orbital geometry and satellite deployment through a Cox point process and by formulating the coverage region as a spherical Cox-Boolean model, the paper establishes a tractable approach to describe the spatial characteristics of heterogeneous constellations.

Using this framework, the connection probability, connection number, and various downlink performance metrics are derived as explicit functions of the system parameters, including the nearest serving distance and the total received signal power, which together provide a geometric characterization of link quality. In addition, a baseline SIR analysis is conducted to illustrate how coverage overlap induced by orbital geometry relates to interference and communication performance. The analytical results are validated through system-level simulations and are further compared with a down-sampled realistic LEO constellation, demonstrating that the proposed model captures the overall trend and scale of key geometric metrics.

Beyond derivations, the developed framework offers practical applicability to the design and evaluation of future satellite infrastructure sharing. The tractable analysis enables the assessment of trade-offs between deployment constraints, coverage overlap, and connectivity reliability. Moreover, our general form of the results allows direct expansion of the analysis to multi-altitude models or arbitrary communication-range distributions, extending its usefulness to various operational environments and system architectures. The framework can thus serve as both a benchmark tool and a design guideline for assessing the performance, scalability, and robustness of shared satellite systems in next-generation NTNs.

\bibliographystyle{IEEEtran}
\bibliography{ref}

\begin{thebibliography}{10}
\providecommand{\url}[1]{#1}
\csname url@samestyle\endcsname
\providecommand{\newblock}{\relax}
\providecommand{\bibinfo}[2]{#2}
\providecommand{\BIBentrySTDinterwordspacing}{\spaceskip=0pt\relax}
\providecommand{\BIBentryALTinterwordstretchfactor}{4}
\providecommand{\BIBentryALTinterwordspacing}{\spaceskip=\fontdimen2\font plus
\BIBentryALTinterwordstretchfactor\fontdimen3\font minus
  \fontdimen4\font\relax}
\providecommand{\BIBforeignlanguage}[2]{{%
\expandafter\ifx\csname l@#1\endcsname\relax
\typeout{** WARNING: IEEEtran.bst: No hyphenation pattern has been}%
\typeout{** loaded for the language `#1'. Using the pattern for}%
\typeout{** the default language instead.}%
\else
\language=\csname l@#1\endcsname
\fi
#2}}
\providecommand{\BIBdecl}{\relax}
\BIBdecl

\bibitem{6934544}
B.~G. Evans, ``The role of satellites in {5G},'' in \emph{Proc. IEEE
  ASMS/SPSC}, 2014, pp. 197--202.

\bibitem{8700141}
Y.~Su, Y.~Liu, Y.~Zhou, J.~Yuan, H.~Cao, and J.~Shi, ``Broadband {LEO}
  satellite communications: Architectures and key technologies,'' \emph{IEEE
  Wireless Commun.}, vol.~26, no.~2, pp. 55--61, 2019.

\bibitem{9275613}
M.~Giordani and M.~Zorzi, ``Non-terrestrial networks in the {6G} era:
  Challenges and opportunities,'' \emph{IEEE Network}, vol.~35, no.~2, pp.
  244--251, 2021.

\bibitem{9861699}
M.~M. Azari, S.~Solanki, S.~Chatzinotas, O.~Kodheli, H.~Sallouha, A.~Colpaert,
  J.~F. Mendoza~Montoya, S.~Pollin, A.~Haqiqatnejad, A.~Mostaani, E.~Lagunas,
  and B.~Ottersten, ``Evolution of non-terrestrial networks from {5G} to {6G}:
  A survey,'' \emph{IEEE Commun. Surveys Tuts.}, vol.~24, no.~4, pp.
  2633--2672, 2022.

\bibitem{9941481}
S.~S. Hassan, D.~H. Kim, Y.~K. Tun, N.~H. Tran, W.~Saad, and C.~S. Hong,
  ``Seamless and energy-efficient maritime coverage in coordinated {6G}
  space–air–sea non-terrestrial networks,'' \emph{IEEE Internet of Things
  J.}, vol.~10, no.~6, pp. 4749--4769, 2023.

\bibitem{8002583}
Z.~Qu, G.~Zhang, H.~Cao, and J.~Xie, ``{LEO} satellite constellation for
  {Internet} of {Things},'' \emph{IEEE Access}, vol.~5, pp. 18\,391--18\,401,
  2017.

\bibitem{FCCKuiper}
\BIBentryALTinterwordspacing
{Federal Communications Commission}. {FCC} authorizes {Kuiper} satellite
  constellation. [Online]. Available:
  \url{https://www.fcc.gov/document/fcc-authorizes-kuiper-satellite-constellation}
\BIBentrySTDinterwordspacing

\bibitem{FCCBoeing}
\BIBentryALTinterwordspacing
------. {FCC} authorizes {Boeing} broadband satellite constellation. [Online].
  Available:
  \url{https://www.fcc.gov/document/fcc-authorizes-boeings-broadband-satellite-constellation}
\BIBentrySTDinterwordspacing

\bibitem{FCCStarlink}
\BIBentryALTinterwordspacing
------. {SpaceX} authorized for {SCS} and operations at lower altitudes.
  [Online]. Available:
  \url{https://www.fcc.gov/document/spacex-authorized-scs-and-operations-lower-altitudes}
\BIBentrySTDinterwordspacing

\bibitem{9371230}
R.~Deng, B.~Di, H.~Zhang, L.~Kuang, and L.~Song, ``Ultra-dense {LEO} satellite
  constellations: How many {LEO} satellites do we need?'' \emph{IEEE Trans.
  Wireless Commun.}, vol.~20, no.~8, pp. 4843--4857, 2021.

\bibitem{9210567}
O.~Kodheli, E.~Lagunas, N.~Maturo, S.~K. Sharma, B.~Shankar, J.~F.~M. Montoya,
  J.~C.~M. Duncan, D.~Spano, S.~Chatzinotas, S.~Kisseleff, J.~Querol, L.~Lei,
  T.~X. Vu, and G.~Goussetis, ``Satellite communications in the new space era:
  A survey and future challenges,'' \emph{IEEE Commun. Surveys Tuts.}, vol.~23,
  no.~1, pp. 70--109, 2021.

\bibitem{okla}
\BIBentryALTinterwordspacing
J.~Fomon. Starlink slowed in {Q2}, competitors mounting challenges. [Online].
  Available:
  \url{https://www.ookla.com/articles/starlink-hughesnet-viasat-performance-q2-2022}
\BIBentrySTDinterwordspacing

\bibitem{9852737}
H.~Al-Hraishawi, H.~Chougrani, S.~Kisseleff, E.~Lagunas, and S.~Chatzinotas,
  ``A survey on nongeostationary satellite systems: The communication
  perspective,'' \emph{IEEE Commun. Surveys Tuts.}, vol.~25, no.~1, pp.
  101--132, 2023.

\bibitem{daley2007introduction}
D.~J. Daley and D.~Vere-Jones, \emph{An Introduction to the Theory of Point
  Processes: Volume {I}: Elementary Theory and Method}.\hskip 1em plus 0.5em
  minus 0.4em\relax Springer, New York, 2003.

\bibitem{chiu2013stochastic}
S.~N. Chiu, D.~Stoyan, W.~S. Kendall, and J.~Mecke, \emph{Stochastic Geometry
  and Its Applications}.\hskip 1em plus 0.5em minus 0.4em\relax John Wiley \&
  Sons, 2013.

\bibitem{baccelli2010stochastic}
F.~Baccelli and B.~B{\l}aszczyszyn, ``Stochastic geometry and wireless
  networks: volume {I} theory,'' \emph{Foundations and Trends® in Networking},
  vol.~3, no. 3--4, pp. 249--449, 2010.

\bibitem{haenggi2012stochastic}
M.~Haenggi, \emph{Stochastic Geometry for Wireless Networks}.\hskip 1em plus
  0.5em minus 0.4em\relax Cambridge University Press, 2012.

\bibitem{11159552}
C.-S. Choi and F.~Baccelli, ``Stochastic geometry and dynamical system analysis
  of {Walker} satellite constellations,'' \emph{early access: IEEE Trans. Veh.
  Technol.}, pp. 1--6, 2025.

\bibitem{6042301}
J.~G. {Andrews}, F.~{Baccelli}, and R.~K. {Ganti}, ``A tractable approach to
  coverage and rate in cellular networks,'' \emph{IEEE Trans. Commun.},
  vol.~59, no.~11, pp. 3122--3134, 2011.

\bibitem{6171996}
H.~S. Dhillon, R.~K. Ganti, F.~Baccelli, and J.~G. Andrews, ``Modeling and
  analysis of {$K$}-tier downlink heterogeneous cellular networks,'' \emph{IEEE
  J. Sel. Areas Commun.}, vol.~30, no.~3, pp. 550--560, 2012.

\bibitem{5226963}
F.~Baccelli, B.~Blaszczyszyn, and P.~Muhlethaler, ``Stochastic analysis of
  spatial and opportunistic {ALOHA},'' \emph{IEEE J. Sel. Areas Commun.},
  vol.~27, no.~7, pp. 1105--1119, 2009.

\bibitem{8340239}
V.~V. Chetlur and H.~S. Dhillon, ``Coverage analysis of a vehicular network
  modeled as {Cox} process driven by {Poisson} line process,'' \emph{IEEE
  Trans. Wireless Commun.}, vol.~17, no.~7, pp. 4401--4416, 2018.

\bibitem{8357962}
C.-S. Choi and F.~Baccelli, ``An analytical framework for coverage in cellular
  networks leveraging vehicles,'' \emph{IEEE Trans. Commun.}, vol.~66, no.~10,
  pp. 4950--4964, 2018.

\bibitem{8419219}
------, ``Poisson {Cox} point processes for vehicular networks,'' \emph{IEEE
  Trans. Veh. Technol.}, vol.~67, no.~10, pp. 10\,160--10\,165, Oct 2018.

\bibitem{9354063}
------, ``Modeling and analysis of vehicle safety message broadcast in cellular
  networks,'' \emph{IEEE Trans. Wireless Commun.}, vol.~20, no.~7, pp.
  4087--4099, 2021.

\bibitem{9079921}
N.~Okati, T.~Riihonen, D.~Korpi, I.~Angervuori, and R.~Wichman, ``Downlink
  coverage and rate analysis of low {Earth} orbit satellite constellations
  using stochastic geometry,'' \emph{IEEE Trans. Commun.}, vol.~68, no.~8, pp.
  5120--5134, 2020.

\bibitem{9177073}
A.~Talgat, M.~A. Kishk, and M.-S. Alouini, ``Nearest neighbor and contact
  distance distribution for binomial point process on spherical surfaces,''
  \emph{IEEE Commun. Lett.}, vol.~24, no.~12, pp. 2659--2663, 2020.

\bibitem{9218989}
------, ``Stochastic geometry-based analysis of {LEO} satellite communication
  systems,'' \emph{IEEE Commun. Lett.}, vol.~25, no.~8, pp. 2458--2462, 2021.

\bibitem{9678973}
D.-H. Jung, J.-G. Ryu, W.-J. Byun, and J.~Choi, ``Performance analysis of
  satellite communication system under the shadowed-{Rician} fading: A
  stochastic geometry approach,'' \emph{IEEE Trans. Commun.}, vol.~70, no.~4,
  pp. 2707--2721, 2022.

\bibitem{kim2024spectrum}
D.~Kim, J.~Park, J.~Choi, and N.~Lee, ``Spectrum sharing between low {Earth}
  orbit satellite and terrestrial networks: A stochastic geometry perspective
  analysis,'' \emph{arXiv preprint arXiv:2408.12145}, 2024.

\bibitem{10901947}
S.~Yang, Y.~Zhu, O.~A. Dobre, G.~K. Karagiannidis, and Z.~Ding, ``Performance
  analysis for {NOMA}-assisted {LEO} communications: A two-dimensional
  stochastic geometric approach,'' \emph{IEEE Trans. Wireless Commun.},
  vol.~24, no.~5, pp. 3822--3836, 2025.

\bibitem{10994488}
R.~Wang, M.~A. Kishk, H.~H. Yang, and M.-S. Alouini, ``Analyzing localizability
  of {LEO/MEO} hybrid networks: A stochastic geometry approach,'' \emph{IEEE
  Trans. Aerosp. Electron. Syst.}, vol.~61, no.~4, pp. 10\,720--10\,736, 2025.

\bibitem{7967745}
V.~V. Chetlur and H.~S. Dhillon, ``Downlink coverage analysis for a finite
  3-{D} wireless network of unmanned aerial vehicles,'' \emph{IEEE Trans.
  Commun.}, vol.~65, no.~10, pp. 4543--4558, 2017.

\bibitem{8713514}
B.~Galkin, J.~Kibiłda, and L.~A. DaSilva, ``A stochastic model for {UAV}
  networks positioned above demand hotspots in urban environments,'' \emph{IEEE
  Trans. Veh. Technol.}, vol.~68, no.~7, pp. 6985--6996, 2019.

\bibitem{8654644}
S.~Zhang, J.~Liu, and W.~Sun, ``Stochastic geometric analysis of multiple
  unmanned aerial vehicle-assisted communications over {Internet} of
  {Things},'' \emph{IEEE Internet Things J.}, vol.~6, no.~3, pp. 5446--5460,
  2019.

\bibitem{8681266}
S.~Enayati, H.~Saeedi, H.~Pishro-Nik, and H.~Yanikomeroglu, ``Moving aerial
  base station networks: A stochastic geometry analysis and design
  perspective,'' \emph{IEEE Trans. Wireless Commun.}, vol.~18, no.~6, pp.
  2977--2988, 2019.

\bibitem{9252889}
C.~K. Armeniakos, P.~S. Bithas, and A.~G. Kanatas, ``{SIR} analysis in {3D}
  {UAV} networks: A stochastic geometry approach,'' \emph{IEEE Access}, vol.~8,
  pp. 204\,963--204\,973, 2020.

\bibitem{9785498}
X.~Shi and N.~Deng, ``Modeling and analysis of {mmWave} {UAV} swarm networks: A
  stochastic geometry approach,'' \emph{IEEE Trans. Wireless Commun.}, vol.~21,
  no.~11, pp. 9447--9459, 2022.

\bibitem{10050345}
M.~Matracia, M.~A. Kishk, and M.-S. Alouini, ``{UAV}-aided post-disaster
  cellular networks: A novel stochastic geometry approach,'' \emph{IEEE Trans.
  Veh. Technol.}, vol.~72, no.~7, pp. 9406--9418, 2023.

\bibitem{10557592}
C.-S. Choi and F.~Baccelli, ``Cox point processes for multi altitude {LEO}
  satellite networks,'' \emph{IEEE Trans. Veh. Technol.}, vol.~73, no.~10, pp.
  15\,916--15\,921, 2024.

\bibitem{10410220}
C.-S. Choi, ``Modeling and analysis of downlink communications in a
  heterogeneous {LEO} satellite network,'' \emph{IEEE Trans. Wireless Commun.},
  vol.~23, no.~8, pp. 8588--8602, 2024.

\bibitem{10703111}
C.-S. Choi and F.~Baccelli, ``A novel analytical model for {LEO} and {MEO}
  satellite networks based on {Cox} point processes,'' \emph{IEEE Trans.
  Commun.}, vol.~73, no.~4, pp. 2265--2279, 2025.

\bibitem{10436110}
C.-S. Choi, ``Analysis of a delay-tolerant data harvest architecture leveraging
  low {Earth} orbit satellite networks,'' \emph{IEEE J. Sel. Areas Commun.},
  vol.~42, no.~5, pp. 1329--1343, 2024.

\bibitem{10771991}
------, ``Leveraging aerial platforms for downlink communications in sparse
  satellite networks,'' \emph{IEEE Internet Things J.}, vol.~12, no.~8, pp.
  9805--9820, 2025.

\bibitem{subCox1}
M.~Ying, X.~Chen, Q.~Qi, and Y.~Xu, ``Modeling and analysis for multiple-layer
  {LEO} satellite internet of things constellations,'' \emph{submitted IEEE
  Trans. Wireless Commun.}, pp. 1--14, 2025.

\bibitem{molchanov2005theory}
I.~S. Molchanov, \emph{Theory of Random Sets}.\hskip 1em plus 0.5em minus
  0.4em\relax Springer, 2005, vol.~87, no.~2.

\bibitem{38821}
{3GPP TR 38.821}, ``Solutions for {NR} to support non-terrestrial networks
  ({NTN}),'' \emph{3GPP TR 38.821}, Apr. 2024.

\bibitem{38901}
{3GPP TR 38.901}, ``{NR}; study on channel model for frequencies from 0.5 to
  100 {GHz},'' \emph{3GPP TR 38.901}, Sep. 2025.

\bibitem{walker1984satellite}
J.~G. Walker, ``Satellite constellations,'' \emph{Journal of the British
  Interplanetary Society}, vol.~37, pp. 559--572, 1984.

\end{thebibliography}

\appendices
\section{Proof of Theorem \ref{T:1}}\label{A:1}
	Based on the definition of the connection probability, it is given by one minus the outage probability, namely the probability that no satellite coverage cell includes the typical user at $u_0$. We hence have
\begin{align*}
	\bP(\texttt{connection}) &= 1- \bP(\texttt{outage})\\
	&=1- \bP\left(u_0\notin \cC_{i,j}\ \forall X_{i,j}\in\Psi \right)\\
	&=1- \bP(\|X_{i,j}-u_0\|>\kappa_{i,j}, \forall X_{i,j} \in \Psi),
\end{align*}
where $\kappa_{i,j}$ denotes the communication range of the satellite $X_{i,j}$, represented as an i.i.d. uniform between $c_m$ and $c_M.$ Since $\kappa_{i,j}$ is i.i.d., conditionally on $\Psi,$ the connection probability is given by
\begin{align}
	p_c & = 1- \bE\left[ \prod_{ X_{i,j} \in \Psi } \bP \left(\left. \kappa_{i,j}\leq \|X_{i,j}-u_0\|\right| \Psi \right)\right], \nnb
\end{align}
where we use the law of total probability.

Then, we employ the definition of the coverage cell that satellites at distances of more than $c_M$ from $u_0$ do not provide the coverage to the typical user. Moreover, for such satellites, $\bP(\kappa_{i,j}\leq \|X_{i,j}-u_0\|) = 1 $. Consequently, the expectation above can be rewritten as follows:
\begin{align}
	&\bE\left[ \prod_{ X_{i,j} \in \Psi }^{\|X_{i,j}-u_0\|<c_M} \bP \left(\left. \kappa_{i,j}\leq \|X_{i,j}-u_0\|\right| \Psi \right)\right]\nnb\\
	&=\bE\left[ \prod_{ X_{i,j} \in \Psi }^{\|X_{i,j}-u_0\|<c_M} F_{\kappa}(\|X_{i,j}-u_0\|)\right]\nnb\\
	&=\bE\left[ \prod_{ O_{i} \in \cO}^{|\phi_i-\frac{\pi}{2}|<\tilde{\varphi}}\left.\bE\left[ \prod_{X_{i,j}\in\psi_i}^{|\omega_{j}-\frac{\pi}{2}|<\bar{\omega}(\phi_i)} F_{\kappa}(\|X_{i,j}-u_0\|)\right|\cO\right]\right]\nnb,
\end{align}
where we let $F_{\kappa}(x)$ be the CDF of the random variable $\kappa$ evaluated at $x$. We used the fact that conditionally on the orbit process $\cO,$ the satellite point processes on different orbits are independent and thus the expectation of the products becomes the product of conditional expectations. Finally, we use the fact that the $j$-th satellites on $i$-th orbit $O(\theta_i,\phi_i) $ at distances less than $c_M$ have their phases between $\pi/2-\bar{\omega}(\phi_{i}) $ and $\pi/2+\bar{\omega}(\phi_{i}).$ Similarly, we also use the fact that the orbits that may contain satellites at distances less than  $c_M$ from the typical user always have their inclination angles between $\pi/2-\tilde{\varphi}$ and $\pi/2+\tilde{\varphi}$. A simple spherical triangle geometry argument gives the expression for $\tilde{\varphi}$ and $\bar{\omega}(\phi_i)$ as follows:
\begin{align}
	&\tilde{\varphi}= \arccos\left({(r_e^2+r_s^2-c_M^2)}/{(2 r_e r_s)}\right),\label{eq:tilde_varphi}\\
	&\bar{\omega}(\phi_i) = \arcsin\left(\sqrt{ 1-\cos^2(\tilde{\varphi})\csc^2(\phi_{i}) }\right)\label{eq:bar_omega},
\end{align}
respectively. Note $\bar{\omega}(\phi_i)$ depends on the inclination angle of the corresponding orbit and hence it is given as a function of $\phi_i.$ Here, $\bar{\omega}(\phi)$ denotes the phase half-width expressed in terms of the inclination variable $\phi$, whereas $\tilde{\omega}(\varphi)$ in Eq. \eqref{eq:tilde omega} is its counterpart after the change of variables $\varphi = \pi/2-\phi$ performed at the end of this proof, i.e., $\tilde{\omega}(\varphi)=\bar{\omega}(\pi/2-\varphi)$, which converts $\csc^2(\phi)$ into $\sec^2(\varphi)$. With a slight abuse of notation, the distance from the satellite $X_{i,j} $ to the typical user, $\|X_{i,j}-u_0\|$, is given by
\begin{align}
	K(\phi_{i},\omega_{j})&= \sqrt{r_s^2-2 r_sr_e \sin(\phi_{i})\sin(\omega_j)+r_{e}^2}\label{eq:K_ij},
\end{align}
where $\phi_{i}$ is the inclination of the $i$-th orbit and $\omega_j$ is the phase of the $j$-th satellite on the $i$-th orbit.

Using the above equations, we obtain the expression for the outage probability as follows:
\begin{align}
	&\bP(\texttt{outage}) \nnb \\
	&= \bE\left[ \prod_{ O_{i} \in \cO}^{|\phi_i-\frac{\pi}{2}|<\tilde{\varphi}}\bE\left[ \prod_{X_{i,j}\in\psi_i}^{|\omega_{j}-\frac{\pi}{2}|<\bar{\omega}(\phi_i)}\left. F_{\kappa}(K(\phi_{i},\omega_{j}))\right|\cO\right]\right]\nnb\\
	&=\bE\left[ \prod_{ O_{i} \in \cO}^{|\phi_i-\frac{\pi}{2}|<\tilde{\varphi}}\exp\left(-\frac{\mu}{2\pi}\int_{\frac{\pi}{2}-\bar{\omega}(\phi_i)}^{\frac{\pi}{2}+\bar{\omega}(\phi_i)}\bar{F}_{\kappa}(K(\phi_i,\omega))\diff \omega\right)\right]\nnb,
\end{align}
where we use the probability generating functional of the Poisson point process of intensity $\mu/2\pi$ on the orbit $ O_i$. Here, we denote by $\bar{F}_\kappa(x)$ the complementary CDF of the random variable $\kappa.$

We then apply the probability generating functional of the Poisson point process $\Xi$ of intensity $\lambda\sin(\phi)/(2\pi)$ on the set $\mathbb{T}=[0,\pi]\times [0,\pi]$ to obtain the expression for $\bP(\texttt{outage})$ as follows:
\begin{align}
	&\exp\left(-\int_{\frac{\pi}{2}-\tilde{\varphi}}^{\frac{\pi}{2}+\tilde{\varphi}}\int_{0}^{\pi}\frac{\lambda\sin(\phi)}{2\pi}\right.\nnb\\
	&\hspace{18mm}	\left.\left(1-e^{-\frac{\mu}{2\pi}\int_{\frac{\pi}{2}-\bar{\omega}(\phi)}^{\frac{\pi}{2}+\bar{\omega}(\phi)}\bar{F}_{\kappa}(K(\phi,\omega))\diff \omega }\right)\diff \theta \diff \phi \right).\nnb
\end{align}

This expression is simplified further using the fact that the outer integral is w.r.t. the variable $\phi $ over the set $[\frac{\pi}{2}-\tilde{\varphi}, \frac{\pi}{2}+\tilde{\varphi}]$ and thus we use the change of variables $\varphi = \pi/2-\phi$ to simplify. Similarly, the inner integration is w.r.t. the variable $\omega$ over the set $[\pi/2-\bar{\omega}(\phi), \pi/2+\bar{\omega}(\phi)]$ and thus we again use the change of variables $\eta = \pi/2-\omega$ to simplify further. Hence, the outage probability is given by
\begin{equation}
	\exp\left({-\int_{0}^{\tilde{\varphi}}{\lambda\cos(\varphi)}\left(1-e^{-\frac{\mu}{\pi}\int_{0}^{\tilde{\omega}(\varphi)}\bar{F}_{\kappa}\left(\tilde{K}\left(\varphi,\eta\right)\right)\diff \eta}\right)\diff \varphi}\right)\nnb,
\end{equation}
where we use two variables as follows:
\begin{align*}
	&\tilde{\omega}(\varphi)= \arcsin\left(\sqrt{ 1-\cos^2(\tilde{\varphi})\sec^2(\varphi) }\right),\nnb\\
	&\tilde{K}(\varphi,\eta)= \sqrt{r_s^2-2 r_sr_e \cos(\varphi)\cos(\eta)+r_{e}^2}.
\end{align*}
This completes the proof.
\section{Proof of Theorem \ref{T:2}}\label{A:2}
	The Laplace transform of the connection number is 
\begin{align}
	\cL_{N_{c}}(t) &=\bE\left[\exp\left(-t\sum_{X_{i,j}\in\Psi}\ind_{u_0\in\cC_{i,j}} \right)\right]\nnb\\
	&= \bE\left[\exp\left(-t\sum_{X_{i,j}\in\Psi}\ind_{\kappa_{i,j}>\|X_{i,j}-u_0\|} \right)\right]\nnb\\
	&=\bE\left[\bE\left[\left.\prod_{ X_{i,j} \in \Psi } \exp\left(-t \ind_{\kappa_{i,j}>\|X_{i,j}-u_0\|}\right)\right| \Psi \right]\right],\nnb
\end{align}
where the summation is now expressed as a product form. Conditionally on $\Psi$, we have
\begin{align}
	\cL_{N_{c}}(t) &= \bE\left[\prod_{ X_{i,j} \in \Psi }^{\|X_{i,j}-u_0\|<c_M} \bE\left[\left. e^{-t \ind_{\kappa_{i,j}>\|X_{i,j}-u_0\|}}\right|\Psi\right]  \right]\nnb\\
	&=\bE\Bigg[\prod_{ X_{i,j} \in \Psi }^{\|X_{i,j}-u_0\|<c_M} \Big(1-(1-e^{-t})\nnb\\
	&\hspace{24mm}\times\bar{F}_{\kappa}(\|X_{i,j}-u_0\|)\Big)\Bigg].
\end{align}
Then, conditionally on the orbit process, the Laplace transform of the connection number is given by
\begin{align}
	&\cL_{N_{c}}(t)\nnb\\
	&=\bE\left[\prod_{O_i\in\cO}\!\bE \left[\left.\prod_{X_{i,j}\in\psi_i}\! \big(1-(1-e^{-t})\bar{F}_{\kappa}(K(\phi_{i},\omega_{j}))\big) \right| \cO\right]\right]\nnb\\
	&=\bE\Bigg[\prod_{O_i\in\cO}^{|\phi_{i}-\frac{\pi}{2}|<\tilde{\varphi}}\!\bE \Bigg[\prod_{X_{i,j}\in\psi_i}^{|\omega_{j}-\frac{\pi}{2}|<\bar{\omega}(\phi_{i})} \!\!\Big(1-(1-e^{-t})\nnb\\
	&\hspace{24mm}\times\bar{F}_{\kappa}(K(\phi_{i},\omega_{j}))\Big)\, \Bigg| \,\cO\Bigg]\Bigg]\nnb,
\end{align}
where $\tilde{\varphi}$  is given by Eq. \eqref{eq:tilde_varphi} and $\bar{\omega}(\phi_{i})$ is given by Eq. \eqref{eq:bar_omega}. We now use the probability generating functional of the Poisson point processes $\Xi$ and $\psi_i$ to obtain the Laplace transform of the connection number as follows:
\begin{align}
	&\cL_{N_{c}}(t)=\exp\left(-\int_{\frac{\pi}{2}-\tilde{\varphi}}^{\frac{\pi}{2}+\tilde{\varphi}}\frac{\lambda\sin(\phi)}{2}\right.\nnb\\
	&\hspace{15mm}\left.\left(1-e^{-\frac{\mu}{2\pi}\int_{\frac{\pi}{2}-\bar{\omega}(\phi)}^{\frac{\pi}{2}+\bar{\omega}(\phi)}(1-e^{-t})\bar{F}_{\kappa}(K(\phi,\omega))\diff \omega}\right)\diff \phi\right).\nnb
\end{align}
Finally, the change of variables gives the final result.
\section{Proof of Theorem \ref{T:3}}\label{A:3}
	Using the definition of the connection number, its expectation is given by
\begin{align}
	n_c & = \bE\left[\sum_{X_{i,j}\in\Psi} \ind_{\kappa_{i,j}>\|X_{i,j}-u_0\|}\right]\nnb\\
	&=\bE\left[\left.\sum_{X_{i,j}\in\Psi}\bE\left[ \ind_{\kappa_{i,j}>\|X_{i,j}-u_0\|}\right.|\Psi\right]\right]\nnb\\
	&=\bE\left[\sum_{X_{i,j}\in\Psi}\bar{F}_{\kappa} (\|X_{i,j}-u_0\|)\right],
\end{align}
where we use the fact that conditionally on the satellite point process $\Psi,$ the distance from the typical user to the satellite points are functions of $\Psi$.
Then, for all $\|X_{i,j}-u_0\|>c_M,$ we have $\bar{F}_{\kappa}(\|X_{i,j}-u_0\|) = 0.$ Hence, we have
\begin{align}
	n_c &= \bE\left[\sum_{O_i\in\cO}\bE\left[\left.\sum_{X_{i,j}\in\psi_{i}}\bar{F}_{\kappa} (\|X_{i,j}-u_0\|)\right|\cO\right]\right]\nnb\\
	&=\bE\left[\sum_{O_i\in\cO}\int_{\frac{\pi}{2}-\bar{\omega}(\phi_i)}^{\frac{\pi}{2}+\bar{\omega}(\phi_i)}\frac{\mu \bar{F}_{\kappa}(K(\phi_i,\omega))}{2\pi}\diff \omega\right]\nnb\\
	&=\int_{\pi/2-\tilde{\varphi}}^{\pi/2+\tilde{\varphi}}\frac{\lambda\sin(\phi)}{2}\int_{\frac{\pi}{2}-\bar{\omega}(\phi)}^{\frac{\pi}{2}+\bar{\omega}(\phi)}\frac{\mu \bar{F}_{\kappa}(K(\phi,\omega))}{2\pi}\diff \omega \diff \phi\nnb\\
	&=\frac{\lambda\mu}{\pi}\int_0^{\tilde{\varphi}}\int_{0}^{\tilde{\omega}(\varphi)}\cos(\varphi)\bar{F}_{\kappa}(\tilde{K}(\varphi,\eta))\diff \eta \diff \varphi,
\end{align}
where we leverage the Campbell formula \cite{baccelli2010stochastic,chiu2013stochastic,daley2007introduction} on the Poisson point processes $\Xi$
and $\psi_{i}$. We then employed the change of variables to obtain the final result.

\section{Proof of Theorem \ref{T:4}}\label{A:4}
	The association distance is the distance from the typical user to its closest satellite out of all satellites whose coverage cells contain the typical user. In addition, the association distance is defined as infinity if there is no satellite that geometrically provides the coverage.

For $c_M\leq z,$ $\bP(Z>z) \equiv \bP(Z=\infty)$ since the maximum communication range for any satellite is assumed to be $c_M.$ Then, the probability of $Z=\infty$ is obtained by one minus the connection probability of the typical user in Theorem \ref{T:1}. Therefore, for $c_M\leq z$, $\bP(Z\leq z) = \bP(\texttt{connection}).$

On the other hand, for $c_m\leq z<c_{M},$ we have
\begin{align}
	\bP(Z\leq z) &= 1-\bP(Z>z).
\end{align}
Let $\bS(z)$ denote the spherical cap of the satellite sphere $\bS$ within distance $z$ from the typical user, i.e., $\bS(z) = \{v\in\bS \mid \|v-u_0\|\leq z\}.$ Then, using the fact that the typical user is associated with the nearest one, we have $\bP(Z>z) = \bP(\Psi^p(\bS(z))=0)$
where $\Psi^{p}$ is a thinned point process based on Definition \ref{D:1}. Then, using the proofs of Theorems \ref{T:1} and \ref{T:2}, we have
\begin{align}
	&\bP(\Psi^{p}(\bS(z))=0)\nnb\\
	& = \bP(\|X_{i,j}-u_0\|>z, \forall X_{i,j} \in\Psi^{p})\nnb\\
	&=\bE\left[\prod_{ O_{i} \in \cO} \prod_{X_{i,j}\in\psi_i^p} \ind_{\|X_{i,j}-u_0\|>z}\right]\nnb\\
	&=\bE\left[\prod_{ O_{i} \in \cO} \bE\left[\left.\prod_{X_{i,j}\in\psi_i^p} \ind_{\|X_{i,j}-u_0\|>z}\right| \cO\right]\right],
\end{align}
where $\psi_{i}^p$ denotes the thinned point process of the satellite point process $\psi_i$ on the orbit $O_i$. We use the fact that the intensity measure of the thinned point process is obtained by the original intensity $\mu/2\pi$ multiplied by the retention probability of a given point at $X_{i,j}$ namely $ \bP(\kappa>K(\phi_i,\omega_j))$, or equivalently, $\bar{F}_{\kappa}(K(\phi_i,\omega_j))$  the CCDF of the random variable $\kappa$ evaluated at $K(\phi_i,\omega_j)$. The above inner expectation is now
\begin{align}
	&\bE\left[\left.\prod_{X_{i,j}\in\psi_i^p} \ind_{\|X_{i,j}-u_0\|>z}\right| \cO\right]\nnb\\
	&=\exp\left(-\frac{\mu}{2\pi}\int_{0}^{2\pi}\ind_{K{(\phi_i,\omega)}\leq z}\bar{F}_{\kappa}(K(\phi_i,\omega))\diff \omega\right)\nnb\\
	&= \exp\left({-{\frac{\mu}{2\pi}\int_{\pi/2-\check{\omega}(z,\phi_i)}^{\pi/2+\check{\omega}(z,\phi_i)}\bar{F}_{\kappa}(K(\phi_i,\omega))}}\diff \omega \right)\nnb,
\end{align}
where we use a simple geometric fact that the integral is over the angle $\{\omega\in[0,2\pi] \text{ such that } K(\phi_{i},\omega)\leq z\}$. The corresponding phases are given by $\pi/2-\check{\omega}(z,\phi_{i})$ and $\pi/2+\check{\omega}(z,\phi_i)$ where we use the following notation:
\begin{align*}
	\check{\omega}(z,\phi_i)&=\arcsin(\sqrt{1-\cos^2(\hat{\varphi}(z))\csc^2(\phi_i)}),\\
	\hat\varphi(z) &=\arccos((r_s^2+r_e^2-z^2)/(2r_sr_e)).
\end{align*}
We apply the probability generating functional of the Poisson point process $\Xi$ of intensity $\lambda\sin(\phi)/2\pi$ on torus $\mathbb{T} $ to have
\begin{align}
	&\bP(\Psi^{p}(\bS(z))=0)\nnb\\
	&=\bE\left[\prod_{ O_{i} \in \cO} \exp\left({-{\frac{\mu}{2\pi}\int_{\pi/2-\check{\omega}(z,\phi_i)}^{\pi/2+\check{\omega}(z,\phi_i)}\bar{F}_{\kappa}(K(\phi_i,\omega))}}\diff \omega \right)\right]\nnb\\
	&=\exp\left(\!-{\lambda}\int_{0}^{\hat{\varphi}(z)}\!\!\!\!\!\!\!\cos(\varphi)\left(1-e^{-\frac{\mu}{\pi}\int_{0}^{\hat{\omega}(z,\varphi)}\bar{F}_{\kappa}(\tilde{K} (\varphi,\eta))\diff \eta}\right)\!\diff \varphi\right),\nnb
\end{align}
where we use the following notation:
\begin{align}
	\hat\varphi(z) &=\arccos((r_s^2+r_e^2-z^2)/(2r_sr_e)),\label{27_1}\\
	\hat{\omega}(z,\varphi)&=\arcsin(\sqrt{1-\cos^2(\hat{\varphi}(z))\sec^2(\varphi)})\label{28},\\
	\tilde{K}(\varphi,\eta)&=\sqrt{r_s^2-2 r_sr_e \cos(\varphi)\cos(\eta)+r_{e}^2}.\label{29}
\end{align}
Finally, for $z< c_m, $ we have $c_m = r_{s} - r_e$ and $\bP(Z>z)=1$ and therefore, $\bP(Z\leq z)=0. $ This completes the proof.

\section{Proof of Theorem \ref{T:5}}\label{A:5}
	Conditionally on orbit point process $\cO$, the average total received power is given by
\begin{align}
	\bE[T]&=\bE\left[\sum_{X_{i,j}\in\Psi^p}\bar{p} H_{i,j}\|X_{i,j}-u_0\|^{-\alpha} \right]\nnb\\
	&=\bE\left[\sum_{O_i\in\cO}\bar{p} \bE\left[\left.\sum_{X_{i,j}\in\psi_i^p} \frac{1}{\|X_{i,j}-u_0\|^{\alpha}}\right|\cO\right]\right]\nnb\\
	&=\bE \left[ \sum_{O_i\in\cO}\left.\bar{p} \int_{\frac{\pi}{2}-\bar{\omega}(\phi_i)}^{\frac{\pi}{2}+\bar{\omega}(\phi_i)}\frac{\mu \bar{F}_{\kappa}(K(\phi_i,\omega))}{2\pi K^{\alpha}(\phi_i,\omega)}\diff \omega\right.\right]\nnb\\
	&= \int_{\frac{\pi}{2}-\tilde{\varphi}}^{\frac{\pi}{2}+\tilde{\varphi}}\!\frac{\lambda\mu \bar{p} \sin(\phi)}{4\pi}\int_{\frac{\pi}{2}-\bar{\omega}(\phi)}^{\frac{\pi}{2}+\bar{\omega}(\phi)}\frac{\bar{F}_{\kappa}(K(\phi,\omega))}{K^{\alpha}(\phi,\omega)}\diff \omega\diff \phi\nnb\\
	&=\int_{0}^{\tilde{\varphi}}\!\frac{\lambda\mu \bar{p} \cos(\varphi)}{\pi}\int_{0}^{\tilde{\omega}(\varphi)}\frac{\bar{F}_{\kappa}(\tilde{K}(\varphi,\eta))}{\tilde{K}^{\alpha}(\varphi,\eta)}\diff \eta\diff \varphi,\nnb
\end{align}
where we use the Campbell formula in \cite{baccelli2010stochastic} twice for the Poisson point processes $\psi_{i}^p$ and then for $\cO$, respectively. Specifically, the inner integral is w.r.t. the thinned Poisson point process $\psi_{i}^p$ and its intensity measure is given by the product of the original intensity parameter $\mu/2\pi$ and the retention probability $\bar{F}_{\kappa}(K(\phi_{i},\omega_{j})) = \bP(\kappa_{i,j}>\|X_{i,j}-u_0\|)=1-F_{\kappa}(K(\phi_i,\omega_{j})).$ Then, we obtain the final result by change of variables. Here, we use the following notation: $K^{\alpha}(\phi,\omega) = (K(\phi,\omega))^{\alpha}.$

On the other hand, conditionally on $\cO$ and $\psi_{i},$ the Laplace transform of the total received power is
\begin{align}
	\cL_{T}(s)
	&=\bE\left[e^{\left.-\sum_{X_{i,j}\in\Psi^p }\frac{s\bar{p}H_{i,j}}{{\|X_{i,j}-u_0\|^{\alpha}}} \right.}\right]\nnb\\
	&=\bE\left[\prod_{ O_{i} \in \cO}\bE\left[\left.\prod_{X_{i,j}\in\psi_i^p}\bE\left[\left.e^{-\frac{s\bar{p}H_{i,j}}{\|X_{i,j}-u_0\|^{\alpha}}}\right|\psi_i^p\right]\right|\cO\right]\right]\nnb\\
	&=\bE\left[\prod_{ O_{i} \in \cO}\bE\left[\left.\prod_{X_{i,j}\in\psi_i^p} \cL_{H}\left(\frac{s\bar{p} }{K^{\alpha}(\phi_{i},\omega_{j})} \right)\right|\cO\right]\right]\nnb,
\end{align}
where we use  $\cL_H(s)$.
Now, we use the probability generating functional of the Poisson point processes \cite{baccelli2010stochastic} to obtain the following expressions:
\begin{align}
	\cL_T(s)
	&= \bE\left[\prod_{ O_{i} \in \cO} e^{-\frac{\mu}{2\pi}\int_{\frac{\pi}{2}-\bar{\omega}(\phi_i)}^{\frac{\pi}{2}+\bar{\omega}(\phi_i)}\bar{F}_{\kappa}(K(\phi_{i},\omega))\left(1-\cL_{H}(s\bar{p}(K(\phi_{i},\omega))^{-\alpha})\right)\diff \omega}\right]\nnb\\
	&=e^{-{\lambda}\int_{0}^{\tilde{\varphi}}\cos(\varphi)\left(1-e^{\left.-\frac{\mu}{\pi}\int_{0}^{\tilde{\omega}(\varphi)}\bar{F}_{\kappa}(\tilde{K}(\varphi,\eta))\left(1-\cL_{H}\left(\frac{s\bar{p}}{\tilde{K}^{\alpha}(\varphi,\eta)}\right)\right)\diff \eta\right.}\right) \diff \varphi}\nnb.
\end{align}
This completes the Laplace transform.

Finally, using the distribution of the normalized association distance in Corollary \ref{C:1}, the average received signal power is
\begin{align}
	\bar{S} &= \bE\left[ \bar{p} H\tilde{Z}^{-\alpha}\right]
	=\bar{p} \left.\int_{0}^{\infty}v^{-\alpha}\frac{\partial \bP(\tilde{Z}\leq v)}{\partial v}\diff v \right.,
\end{align}
where we use the independence between $\tilde{Z}$ and $H$. Then, we leverage Corollary \ref{C:1} to obtain the final result.

\end{document}